\documentclass[journal,twoside,web]{ieeecolor}
\usepackage{generic}
\usepackage{amsmath,amssymb,amsfonts}
\usepackage{graphicx}
\usepackage{hyperref}
\usepackage{textcomp}

\usepackage{subcaption}
\usepackage{color}
\usepackage{diagbox}
\usepackage{adjustbox}
\usepackage{cite}
\usepackage{makecell}
\usepackage{multirow}
\usepackage[noend]{algpseudocode}
\usepackage[T1]{fontenc}
\usepackage{calrsfs}
\DeclareMathAlphabet{\pazocal}{OMS}{zplm}{m}{n}

\DeclareMathOperator*{\argmin}{arg\,min}

\newcommand{\R}{\mathbb{R}}
\newcommand{\Z}{\mathbb{Z}}
\newcommand{\Sym}{\mathbb{S}}
\newcommand{\E}{\mathbb{E}}

\newcommand{\Ltwo} {{\pazocal{L}_2} }

\newcommand{\de}{\mathrm{d}}
\newcommand{\vect}{\mathrm{vec}}
\newcommand\norm[1]{\| #1 \|}
\newcommand\normF[1]{\| #1 \|_\mathrm{F}}
\newcommand{\diag}{\mathrm{diag}}
\newcommand{\Tr}{\mathrm{Tr}}

\newcommand\innprod[2]{\langle #1, #2 \rangle}

\newcommand\eigmax[1]{\lambda_{\mathrm{max}}\left(#1\right)}
\newcommand\eigmin[1]{\lambda_{\mathrm{min}}\left(#1\right)}

\newtheorem{defn}{Definition}
\newtheorem{assum}{Assumption}
\newtheorem{rem}{Remark}
\newtheorem{lem}{Lemma}
\newtheorem{thm}{Theorem}
\newtheorem{cor}{Corollary}

\def\BibTeX{{\rm B\kern-.05em{\sc i\kern-.025em b}\kern-.08em
    T\kern-.1667em\lower.7ex\hbox{E}\kern-.125emX}}
\begin{document}
\title{LQR for Systems with Probabilistic Parametric Uncertainties: A Gradient Method}
\author{Leilei Cui, \IEEEmembership{Member, IEEE}, and Richard D. Braatz, \IEEEmembership{Fellow, IEEE}
\thanks{This work was supported by the U.S. Food and Drug Administration under the FDA BAA-22-00123 program (Award Number 75F40122C00200). 
}
\thanks{L. Cui is with the University of New Mexico, Albuquerque, NM 87131, USA (e-mail: lcui@unm.edu). }
\thanks{R. D. Braatz is with the Massachusetts Institute of Technology, Cambridge, MA 02139, USA (e-mail: braatz@mit.edu).}
}

\maketitle

\begin{abstract}
A gradient-based method is proposed for solving the linear quadratic regulator (LQR) problem for linear systems with nonlinear dependence on time-invariant probabilistic parametric uncertainties. The approach explicitly accounts for model uncertainty and ensures robust performance. By leveraging polynomial chaos theory (PCT) in conjunction with policy optimization techniques, the original stochastic system is lifted into a high-dimensional linear time-invariant (LTI) system with structured state-feedback control. A first-order gradient descent algorithm is then developed to directly optimize the structured feedback gain and iteratively minimize the LQR cost. We rigorously establish linear convergence of the gradient descent algorithm and show that the PCT-based approximation error decays algebraically at a rate $O(N^{-p})$ for any positive integer $p$, where $N$ denotes the order of the polynomials. Numerical examples demonstrate that the proposed method achieves significantly higher computational efficiency than conventional bilinear matrix inequality (BMI)-based approaches.
\end{abstract}


\section{Introduction}
\label{sec:introduction}
\IEEEPARstart{R}{einforcement} learning (RL) is a dynamic branch of machine learning that optimizes cumulative cost through continuous interaction with the environment. Policy optimization (PO) has been a driving force behind the success of RL, particularly in domains such as video games and board games \cite{mnih2015human, silver2017mastering}. PO typically involves two key steps: parameterizing the control policy using universal approximations, such as neural networks, and updating the parameters by following the gradient descent direction of the cost function \cite[Chapter 13]{book_sutton}. The linear quadratic regulator (LQR), first introduced in Kalman's seminal work \cite{kalman1960contributions}, is a theoretically elegant control method that has seen widespread application across various engineering fields. Consequently, LQR has long served as a benchmark problem for PO in control theory, and has recently experienced renewed interest due to advancements in RL \cite{Levine1970, makila1987computational, fazel2018global, mohammadi2021convergence, hu2023toward}. Since PO is a direct method for controller design, it can be extended to other LQR-related problems, such as static output-feedback control \cite{Levine1970, fatkhullin2021optimizing}, decentralized LQR \cite{li2021distributed}, and linear quadratic Gaussian control \cite{tang2023analysis}. The robustness of PO for LQR under uncertainties has also been studied within the framework of input-to-state stability \cite{Cui2024,cui2023lyapunov}, which implies that the PO can still converge to a neighborhood of the optimal policy even in the presence of noise.

The control policy generated by the standard LQR cannot ensure the stability of the closed-loop system in the presence of model uncertainties, which has motivated researchers to develop PO algorithms for $\pazocal{H}_2/\pazocal{H}_\infty$ control \cite{zhang2021policy, Cui2024RobustRL, guo2022global}. In robust control, the most common approach to handling parametric uncertainties is through worst-case controller design \cite{petersen2014robust}. However, worst-case robust control can be overly conservative, as it assumes the system operates under the worst possible uncertainties, even though such conditions may occur with low or negligible probability. Furthermore, most robust control techniques require the uncertain parameters to be norm-bounded, affine, or represented as a convex polytope \cite{polyak2021linear,petersen2014robust}, limiting their applicability to systems with nonlinear dependence on uncertain parameters unless such parameters are overbounded. 

In contrast, describing parameters with probability distributions allows control designs to capture the likelihood of different uncertainties, providing a more balanced trade-off between performance and robustness. This approach leads to the development of probabilistic robustness \cite{Tempo1996,polyak2001probabilistic,boyarski2005robust,tempo2013randomized}, where uncertainties are treated probabilistically rather than deterministically. Recently, domain randomization and Monte Carlo (MC) sampling are used for PO of LQR \cite{fujinami2025policygradientlqrdomain}. In these methods, randomized algorithms based on MC sampling are used to address probabilistic uncertainty, but they tend to be computationally demanding due to the need for a large number of parameter samples to ensure accuracy. Since the pioneering work of Wiener \cite{wiener1938homogeneous}, polynomial chaos theory (PCT) has been employed to approximate the original model by a deterministic model with augmented states. PCT has been validated through practical applications as a computationally efficient method for uncertainty quantification and propagation in both open- and closed-loop systems \cite{Kim2013,mishra2024polynomial}. By utilizing the PCT, many control problems, such as stability analysis, LQR, $\pazocal{H}_2$ control, and $\pazocal{H}_\infty$ control, have been studied for systems with probabilistic parameters \cite{hover2006application,wan2021polynomial,wan2022polynomial,fisher2009linear,bhattacharya2019robust,lucia2017stability,nechak2021stochastic,evangelisti2024application}. Due to the structural constraints imposed on the feedback gain after PCT-based approximation, all of the existing methods are formulated as optimization problems involving bilinear matrix inequalities (BMIs). As highlighted in \cite{Toker1995,vanantwerp2000tutorial}, solving BMI-based optimization problems is NP-hard and can become computationally infeasible as the problem size increases. After PCT-based approximation, the dimension of the systems is typically high, further exacerbating these challenges, which can lead to either failure in finding a solution or excessive computational complexity, particularly for large-scale systems.

This paper proposes a computationally efficient algorithm for solving optimal control problems for linear systems with probabilistic parametric uncertainties. PCT is employed to transform the original stochastic system into a surrogate deterministic model with augmented, high-dimensional states. Based on this surrogate model, we derive the gradient of the LQR cost function and apply a gradient descent algorithm to directly optimize the structured state-feedback gain. We rigorously establish linear convergence of the gradient descent algorithm and prove that the PCT-based approximation error decays algebraically at a rate of $O(N^{-p})$ for any positive integer $p$, where $N$ denotes the polynomial order. In contrast, the approximation error of Monte Carlo-based domain randomization methods typically decays at the slower rate $O(N^{-1/2})$. In addition, the proposed policy optimization algorithm requires solving only two Lyapunov equations per iteration, resulting in significantly lower computational complexity compared to conventional BMI-based approaches, which typically require solving a sub-optimization problem at each iteration. The main contributions of this paper are threefold:
\begin{itemize}
    \item A computationally efficient policy optimization algorithm is developed for solving the LQR problem of systems with probabilistic parameters.
    \item Linear convergence of the gradient descent algorithm and algebraic convergence of the PCT-based approximation error are rigorously established.
    \item Numerical simulations validate the proposed approach and demonstrate its superior efficiency compared with conventional robust control methods and BMI solvers.
\end{itemize}

The rest of this article is organized as follows. Section \MakeUppercase{\romannumeral 2} summarizes PCT, and the LQR problem for systems with probabilistic parameters is formulated. Section \MakeUppercase{\romannumeral 3} develops a surrogate model for the original stochastic system and proposes a PO algorithm for the surrogate model. Numerical simulations in Section \MakeUppercase{\romannumeral 4} validate the effectiveness and computational efficiency of the proposed algorithm. Finally, the article is concluded in Section \MakeUppercase{\romannumeral 5}.

\section{Preliminaries and Problem Statement}
\subsection{Notations}
$\R$ ($\R_+$) denotes the set of (nonnegative) real numbers, $\mathbb{C}$ is the set of complex numbers, $\Z$ ($\Z_+$) is the set of (nonnegative) integers, and $\Sym ^n$ ($\Sym^n_+$/$\Sym^n_{++}$) is the set of $n$-dimensional real symmetric (positive semidefinite/definite) matrices. $I_n$ is the $n$-dimensional identity matrix. $\Tr(\cdot)$ is the trace of a square matrix. $\eigmax{A}$ and $\eigmin{A}$ are the eigenvalues of $A$ with the largest and smallest real part, respectively. $\norm{\cdot}$ is the Euclidean norm of a vector or the spectral norm of a matrix. $\norm{\cdot}_F$ is the Frobenius norm of a matrix. For any $A,B\in\R^{m \times n}$, their inner product is $\innprod{A}{B} = \Tr[A^\top B]$. $\ell_2(\Z_+, \R)$ denotes the space of square-summable sequences equipped with the norm $\norm{\cdot}_{\ell_2}$. $\Ltwo(\Omega, \pazocal{F}, \mathbb{P})$ is the probability space with finite second moment equipped with the norm $\norm{\cdot}_{\Ltwo}$. For a random matrix $A \in \Ltwo(\Omega, \pazocal{F}, \mathbb{P})$, its norm is $\norm{A}_\Ltwo^2 =\E[\norm{A}^2] = \int_{\Omega} \norm{A}^2 \de \mathbb{P}$. For a matrix $A \in \R^{m \times n}$, $\vect(A) = [a_1^\top,a_2^\top,\cdots{},a_n^\top]^\top$, where $a_i$ is the $i$th column of $A$. $\otimes$ denotes the Kronecker product. $A \succ B$ ($A \succeq B$) indicates that $A-B$ is positive definite (semidefinite). Let $\pazocal{A}\in\mathbb{R}^{(N+1)n\times (N+1)m}$ be partitioned into an $(N+1)\times (N+1)$ block matrix with blocks of size $n\times m$. Denote by $[\pazocal{A}]_{i,j}$ the $(i,j)$th block, by $[\pazocal{A}]_{i,l:s}$ the block row submatrix formed by the blocks in row $i$ and columns $l$ through $s$, and by $[\pazocal{A}]_{l:s,j}$ the block column submatrix formed by the blocks in rows $l$ through $s$ and column $j$.

\subsection{Preliminaries of Orthonormal Polynomials and Polynomial Chaos Expansion}
Some basic aspects of orthonormal polynomials and polynomial chaos theory (PCT) are reviewed in this subsection. Throughout the paper, $\xi \in \mathbb{R}^{n_\xi}$ denotes a random vector defined on the probability space $(\Omega,\pazocal{F},\mathbb{P})$ with compact support $\Xi$. A collection of polynomials $\{\varphi_i(\xi)\}_{i\in\mathbb{Z}_+}$, where each $\varphi_i:\mathbb{R}^{n_\xi}\to\mathbb{R}$, is said to be orthonormal if $\varphi_0(\xi)\equiv 1$ and
\begin{align}
\begin{split}
    \int_{\Xi} \varphi_i(\xi)\varphi_j(\xi),\de\mathbb{P}
    =
\mathbb{E}\!\left[\varphi_i(\xi)\varphi_j(\xi)\right]
    =
    \begin{cases}
    1, & i=j,\\
    0, & i\neq j.
    \end{cases}
\end{split}
\end{align}
For example, Legendre polynomials are defined on $[-1,1]$ and orthonormal with respect to the uniform distribution and the Hermite polynomials are defined on $\R$ and orthonormal with respect to the normal distribution \cite{Xiu2010,book_Chihara}.

For any functions $\psi(\xi): \R^{n_\xi} \to \R$ with a finite second-order moment, i.e. $\psi \in \Ltwo(\Omega, \pazocal{F}, \mathbb{P})$, its polynomial chaos expansion (PCE) is
\begin{align}\label{eq:PCEIntro}
    \psi(\xi) = \sum_{i=0}^\infty \rho_i\varphi_i(\xi) = \rho^\top \phi(\xi)
\end{align}
where $\phi(\xi) = [\varphi_0(\xi), \varphi_1(\xi), \cdots]^\top$, $\rho = [\rho_0, \rho_1, \cdots]^\top \in \ell_2(\Z_+, \R)$, $\{\rho_i\}_{i=0}^\infty$ are the coefficients of the expansion, which can be computed by projecting $\psi(\xi)$ onto the direction of $\varphi_i(\xi)$, i.e.
\begin{align}\label{eq:Galerkin}
    \rho_i = \E[\psi(\xi) \varphi_i(\xi)].
\end{align}
By using the orthogonality of $\varphi_i(\xi)$, the Parseval identity 
\begin{align}\label{eq:ParsevalEq}
    \norm{\psi(\xi)}_\Ltwo^2 = \sum_{i=0}^\infty \rho_i^2 = \norm{\rho}_{\ell_2}^2,
\end{align}
can be obtained.

The $N$th-order PCE approximation of $\psi(\xi)$ can be represented as 
\begin{align}
    \hat{\psi}_{N}(\xi) = \E[ \psi(\xi)\phi_{N}(\xi)]^\top \phi_{N}(\xi)
\end{align}
where $\phi_{N}(\xi) = [\varphi_0(\xi),\cdots{},\varphi_N(\xi)]^\top$ is the vector of $N$th-order orthonormal polynomials. If the orthonormal polynomials are chosen from the Askey scheme according to the distribution of $\xi$, the approximation error in \eqref{eq:PCEIntro} exponentially converges to zero in the $\pazocal{L}_2$ sense \cite{Xiu2010,Xiu2002SIAM}.

\subsection{LQR Control for Systems with Probabilistic Parametric Uncertainty}
Consider the linear time-invariant (LTI) dynamical system,
\begin{align}\label{eq:LTIsys}
\begin{split}
    \dot{x}(t,\xi) &= A(\xi)x(t,\xi) + B(\xi)u(t), \\
    x(0,\xi) &= x_0, 
\end{split}
\end{align}    
where $\xi \in \Xi \subset \R^{n_\xi}$ is the random parameter with known probability density function, $\Xi$ is the compact support of $\xi$, $x(t,\xi) \in \R^{n_x}$ is the system state, $u(t) \in \R^{n_u}$ is the control input, $x_0 \sim \pazocal{N}(0,I_{n_x})$ is the initial state drawn from a standard Gaussian distribution, and $A(\xi)$ and $B(\xi)$ are analytic functions of $\xi$ with compatible dimensions. At each fixed time $t$, the state $x(t,\xi)$ is a function of $\xi \in \Xi$, instead of a finite-dimensional vector, and hence, the system is inherently infinite dimensional. For this system depends on the probabilistic parameters, introduce the concept of uniform exponential (UE) stability. 

\begin{defn}\label{def:UniformExpsta}
    A system $\dot{x}(t,\xi) = A(\xi)x(t,\xi)$ is {\em uniformly exponentially stable} if there exist constants $\mu > 0$ and $\gamma >0$ such that
    \begin{align}\label{eq:MSEstable}
        \norm{x(t,\xi)} \le \mu e^{-\gamma t} \norm{x_0}
    \end{align}
    for all $x_0\in \R^n$, all $\xi \in \Xi$, and all $t \ge 0$.
\end{defn}
The following assumption is provided to guarantee that the system can be UE stabilized by a static state-feedback controller $u(t) = -Kx(t,\xi)$.
\begin{assum}\label{ass:UEstabilizable}
    There exists a feedback gain $K \in \R^{n_u \times n_x}$ such that the closed-loop system $\dot{x}(t,\xi) = (A(\xi) - B(\xi)K)x(t,\xi)$ is UE stable. 
\end{assum}

The closed-loop state matrix corresponding to the feedback gain $K$ is defined as
\begin{align}
A_c(K,\xi) := A(\xi) - B(\xi)K .
\end{align}
The set of all the UE-stabilizing feedback gains is denoted by 
\begin{align}\label{eq:admissibleSet1}
    \pazocal{S} = \{K \in \R^{n_u \times n_x}\ | \ A_c(K,\xi) \text{ is UE stable}\}.
\end{align} 
which is said to be {\em the admissible set of feedback gains}. A feedback gain $K$ is called admissible if it belongs to $\pazocal{S}$.

For the conventional LQR control of LTI systems, the parameter $\xi$ is deterministic. For any fixed $\xi \in \Xi$ and $x_0 \in \R^{n_x}$, LQR aims to find an optimal controller such that the cumulative quadratic cost is minimized:
\begin{align}
\begin{split}
    &\min_u\pazocal{J}_\xi(u,x_0) = \int_{0}^\infty x(t,\xi)^\top Q x(t,\xi) + u(t)^\top R u(t) \de t ,
\end{split}
\end{align}
where $Q \in \Sym^{n_x}_{++}$ and $R \in \Sym_{++}^{n_u}$. For any fixed $\xi \in \Xi$, $K \in \pazocal{S}$, and $x_0 \in \R^{n_x}$, according to \cite{Kleinman1968}, the cost of the closed-loop system with the controller $u(t) = -Kx(t,\xi)$ is 
\begin{align}\label{eq: x0LQR}
    \pazocal{J}_\xi(K,x_0) = x_0^\top P(K,\xi) x_0,
\end{align}
where $P(K, \xi) \in \Sym^n_{++}$ is the solution of the Lyapunov equation
\begin{align}\label{eq:LyapuP}
    &A_c(K,\xi)^\top  P(K,\xi) +  P(K,\xi) A_c(K,\xi) + K^\top R K + Q = 0.
\end{align}
LQR control for systems with probabilistic parameters aims to minimize the expectation of the cost over the parameter $\xi$, that is,
\begin{align}\label{eq:cost}
    \min_{K \in \pazocal{S}} \pazocal{J}(K) = \E[ \pazocal{J}_\xi(K, x_0)] = \E[ \Tr(P(K,\xi))].
\end{align}

To facilitate the analysis of the results below, the following Lyapunov equation is introduced:
\begin{align}\label{eq:LyapuY}
    &A_c(K,\xi)  Y(K,\xi) +  Y(K,\xi) A_c(K,\xi)^\top + I_{n_x} = 0.    
\end{align}
\section{Policy Optimization Using Surrogate Model}
Due to the infinite-dimensional nature of the system in \eqref{eq:LTIsys}, the expectation in \eqref{eq:cost} must be computed over an infinite set of systems, posing a significant challenge to the solvability of the LQR problem for systems with probabilistic parametric uncertainty. To address this, we employ PCT to derive a finite-dimensional approximate model. Subsequently, a gradient-based method---more computationally efficient than conventional BMI-based approaches---is introduced to solve the LQR problem.

\subsection{Finite-Dimensional Approximation Using PCT}

Using PCT, $x(t,\xi)$ can be estimated by
\begin{align}\label{eq:xPCE}
\begin{split}
     x(t,\xi) &= Z_{N}(t) \phi_{N}(\xi) + e_{N}(t,\xi),
\end{split} 
\end{align}
where $Z_{N}(t) \in \R^{n_x \times (N+1)}$ is the coefficients of the orthonormal polynomials, $\phi_{N}(\xi) \in \R^{N+1}$ is the vector of $N$th-order orthonormal polynomials, and $e_{N}(t,\xi) \in \R^{n_x}$ is the residual errors. The coefficients are calculated by Galerkin projection in \eqref{eq:Galerkin}, that is,
\begin{align}\label{eq:ZNprojection}
    Z_{N}(t) = \E [x(t,\xi) \phi_N(\xi)^\top].
\end{align}
In addition, the initial condition $x_0$ can be represented by 
\begin{align}
\begin{split}
    x_0 &=  Z_{N}(0) \phi_{N}(\xi) \\
    Z_{N}(0) &= \E [x_0 \phi_N(\xi)^\top] = [x_0, 0,\cdots{},0].
\end{split}
\end{align}

Considering that $x(t,\xi)$ is the solution of \eqref{eq:LTIsys} with $u(t) = -Kx(t,\xi)$, and using \eqref{eq:xPCE} and \eqref{eq:ZNprojection} yield 
\begin{align}\label{eq:Zdot}
\begin{split}
    \dot{Z}_{N}(t) &= \E [A_c(K,\xi) Z_{N}(t)\phi_N(\xi) \phi_N(\xi)^\top] \\
    &\quad + \E [A_c(K,\xi) e_{N}(t,\xi) \phi_N(\xi)^\top].
\end{split}
\end{align}
Denoting $z_N(t) = \vect(Z_{N}(t))$ and vectoring \eqref{eq:Zdot} yield 
\begin{align}\label{eq:zdot}
    \dot{z}_N(t) &= (\pazocal{A}_N-\pazocal{B}_N(I_{N+1} \otimes K))z_N(t) + \pazocal{R}_N(t),
\end{align}    
where 
\begin{align}\label{eq:LTIPCEABR}
\begin{split}
    \pazocal{A}_N &= \E[ \phi_N(\xi)\phi_N(\xi)^\top \otimes A(\xi)], \\
    \pazocal{B}_N &= \E[ \phi_N(\xi)\phi_N(\xi)^\top \otimes B(\xi)], \\
    \pazocal{R}_N(t) &= \E\!\big[\vect[A_c(K,\xi) {e}_{N}(t,\xi)\phi_N(\xi)^\top]\big].
\end{split}    
\end{align}
The term $\pazocal{R}_N(t)$ in \eqref{eq:LTIPCEABR} is induced by the approximation residual error $e_N(t,\xi)$. When the number of orthonormal polynomials is large enough that $\pazocal{R}_N(t)$ is neglectable, the surrogate LTI system is
\begin{align}\label{eq:surrogateLTI}
\begin{split}
    \dot{\hat{z}}_N(t) &= (\pazocal{A}_N-\pazocal{B}_N(I_{N+1} \otimes K))\hat{z}_N(t),  \, \hat{z}_N(0) = \pazocal{I}_N x_0,
\end{split}
\end{align}
where 
$$\pazocal{I}_N = [I_{n_x}, 0_{n_x \times n_x}, \cdots{}, 0_{n_x \times n_x}]^\top .
$$
System \eqref{eq:surrogateLTI} is a surrogate model of \eqref{eq:LTIsys}. Now, we approximate the infinite-dimensional system in \eqref{eq:LTIsys} by a finite-dimensional system \eqref{eq:surrogateLTI} with an augmented state $\hat{z}_N(t) \in R^{(N+1)n_x}$. The state $x(t,\xi)$ of the original system \eqref{eq:LTIsys} can be estimated as
\begin{align}
    \hat{x}(t,\xi) = \hat{Z}(t)\phi_N(\xi) = (\phi_N(\xi) \otimes I_{n_x})^\top \hat{z}_N(t)
\end{align}
The closed-loop state matrix of the surrogate model under the feedback gain $K$ is given by
\begin{align}
\pazocal{A}_{N,c}(K) := \pazocal{A}_{N} - \pazocal{B}_{N} \bigl(I_{N+1} \otimes K\bigr) .
\end{align}
The set of all the stabilizing feedback gains for the surrogate model is
\begin{align}
    \hat{\pazocal{S}}_N = \left\{K \in \R^{n_u \times n_x}| \pazocal{A}_{N,c}(K) \text{ is Hurwitz} \right\}.
\end{align}

By applying the surrogate model, the averaged LQR cost in \eqref{eq:cost} is approximated by
\begin{align}\label{eq:surrogateJ}
\begin{split}
    &\hat{\pazocal{J}}_N(K) = \int_{0}^\infty \hat{x}(t,\xi)^\top(Q+K^\top RK)\hat{x}(t,\xi) \de t\\
    &= \E \Big[\int_{0}^\infty \hat{z}_N(t)^\top \Phi_N(\xi)  (Q + K^\top R K) \Phi_N(\xi)^\top \hat{z}_N(t) \mathrm{d}t \Big] \\
    &=\E \Big[   \int_{0}^\infty \hat{z}_N(t)^\top(I_{N+1} \otimes (Q + K^\top R K)) \hat{z}_N(t) \de t \Big] \\
    &=\E \big[  x_0^\top \pazocal{I}_N^\top {\pazocal{P}}_N(K) \pazocal{I}_N x_0  \big]
\end{split}
\end{align}
where ${\pazocal{P}}_N(K) \in \Sym^{(N+1)n_x}_{++}$ is the solution of
\begin{align}\label{eq:LyapunovPPCE}
\begin{split}
    \pazocal{A}_{N,c}(K)^\top {\pazocal{P}}_N(K) &+ {\pazocal{P}}_N(K) \pazocal{A}_{N,c}(K) \\
    &+ I_{N+1} \otimes (Q + K^\top RK) = 0 .   
\end{split}
\end{align}
Since $x_0 \sim \pazocal{N}(0,I_{n_x}) $, the surrogate policy optimization problem can then be summarized as
\begin{align}\label{eq:surrogateOpt}
    \min_{K \in \hat{\pazocal{S}}_N} \hat{\pazocal{J}}_N(K) = \Tr[{\pazocal{P}}_N(K)\pazocal{I}_N \pazocal{I}_N^\top].
\end{align}
As established in a subsequent section, the approximation error between $\hat{\pazocal{J}}_N$ and $\pazocal{J}$ decays at the rate $O(N^{-p})$ for any $p\in\mathbb{Z}+$. The variable ${\pazocal{Y}}_N(K) \in \Sym_{++}^{(N+1)n_x}$ is introduced to assist the derivation of the gradient:
\begin{align}\label{eq:LyapunovYPCE}
\begin{split}
    &\pazocal{A}_{N,c}(K) {\pazocal{Y}}_N(K) + {\pazocal{Y}}_N(K) \pazocal{A}_{N,c}(K)^\top + \pazocal{I}_N \pazocal{I}_N^\top = 0.
\end{split}
\end{align}  
By Corollary \ref{cor:LyapPairtrace}, the objective function \eqref{eq:surrogateOpt} can be rewritten as
\begin{align}\label{eq:costJbyY}
    \hat{\pazocal{J}}_N(K) = \Tr[(I_{N+1} \otimes (Q+K^\top RK)){\pazocal{Y}}_N(K)).
\end{align}

\begin{rem}
    The structural constraint on the state-feedback gain formulated as $\pazocal{K}_N = I_{N+1} \otimes K$ hinders the convexification of \eqref{eq:surrogateOpt}. The existing literature, for example \cite{fisher2009linear,wan2021polynomial}, formulates \eqref{eq:surrogateOpt} as an optimization over BMIs, where the objective function is $\hat{\pazocal{J}}_N(K)$ and the BMI constraint is \eqref{eq:LyapunovPPCE} with $``="$ replaced by $``\preceq"$. As pointed out in \cite{Toker1995,vanantwerp2000tutorial}, solving BMIs is NP-hard, which implies that it is highly unlikely that there exists a polynomial-time algorithm for solving BMIs. This limits the existing methods to systems with modest size. In the following, we propose a direct gradient descent algorithm to solve \eqref{eq:surrogateOpt} without resorting to BMIs.
\end{rem}

\subsection{Gradient Computation for the Surrogate Objective Function}

The smoothness of $\hat{\pazocal{J}}_N(K)$ over the set $\hat{\pazocal{S}}_N$ can be demonstrated by \eqref{eq:LyapunovPPCE}. Indeed, vectorizing \eqref{eq:LyapunovPPCE} results in
\begin{align}
\begin{split}
    &\vect\!\big({\pazocal{P}}_N(K)\big)  = -\big[I_{(N+1)n_x} \otimes \pazocal{A}_{N,c}(K)^\top\\
    &+\, \pazocal{A}_{N,c}(K)^\top \otimes I_{(N+1)n_x}\big]^{-1}   \vect\!\big( I_{N+1} \otimes (Q + K^\top RK) \big),
\end{split}
\end{align}
which shows that ${\pazocal{P}}_N(K)$ is an analytic function of $K$. The gradient of $\hat{\pazocal{J}}_N(K)$ is derived in the below lemma.

\begin{lem}\label{lm:gradientPCE}
For any $K \in \hat{\pazocal{S}}_N$, the gradient of $\hat{\pazocal{J}}_N(K)$ is
\begin{align}\label{eq:gradientPCE}
\begin{split}
    \nabla \hat{\pazocal{J}}_N(K) &= 2 \Big[RK \sum_{i=0}^{N}[\pazocal{Y}_N(K)]_{i,i} \\
    &\quad - \sum_{i,j=0}^N [\pazocal{H}_N(K)]_{i,j}[\pazocal{Y}_N(K)]_{j,i}\Big]
\end{split}
\end{align}
where $\pazocal{H}_N(K) = \pazocal{B}_N^\top {\pazocal{P}}_N(K)$.
\end{lem}
\begin{proof}
    The increment of Lyapunov equation \eqref{eq:LyapunovPPCE} is
    \begin{align}\label{eq:IncrementLyapunovPPCE}
    \begin{split}
        0 &= \pazocal{A}_{N,c}(K)^\top \de {\pazocal{P}}_N(K) + \de {\pazocal{P}}_N(K) \pazocal{A}_{N,c} \\
        &\, + (I_{N+1} \otimes \de K^\top)\pazocal{E}_N(K) +\pazocal{E}_N(K)^\top(I_{N+1} \otimes \de K), 
    \end{split}
    \end{align}
    where $\pazocal{E}_N(K)$ is defined as 
    \begin{align}
        \pazocal{E}_N(K) := (I_{N+1} \otimes RK - \pazocal{B}_N^\top{\pazocal{P}}_N(K)).
    \end{align}
    Hence, by using Corollary \ref{cor:LyapPairtrace} and \eqref{eq:LyapunovYPCE}, it holds that
    \begin{align}\label{eq:gradientPCE1}
    \begin{split}
        &\de \hat{\pazocal{J}}_N(K) = \Tr[\de {\pazocal{P}}_N(K) \pazocal{I}_N \pazocal{I}_N^\top]\\
        &= 2\Tr\Big[(I_{N+1} \otimes \de K^\top)\pazocal{E}_N(K){\pazocal{Y}}_{N}(K)\Big] \\
        &= 2 \Tr\!\Big[\de K^\top\sum_{i=0}^{N}RK [\pazocal{Y}_N(K)]_{i,i} \Big] \\
        &\quad - 2 \Tr\!\Big[\de K^\top\sum_{i,j=0}^{N}[\pazocal{H}_N(K)]_{i,j} [\pazocal{Y}_N(K)]_{j,i} \Big].
    \end{split}    
    \end{align}
    Therefore, \eqref{eq:gradientPCE} readily follows from \eqref{eq:gradientPCE1} by noting that $\de \hat{\pazocal{J}}_N(K) = \Tr[\de K^\top \nabla \hat{\pazocal{J}}_N(K)]$.
\end{proof}

With the expression of the gradient in Lemma \ref{lm:gradientPCE}, the first-order gradient descent method can be applied to optimize the surrogate objective function, which is
\begin{align}\label{eq:gradientDescent}
    K_{k+1} = K_{k} -\eta_k \nabla \hat{\pazocal{J}}_N(K_k),
\end{align}
where $\eta_k> 0$ is the step size of the gradient descent algorithm. The convergence of \eqref{eq:gradientDescent} to the global minimum is analyzed in the next section.

\subsection{Hessian Matrix Computation for the Surrogate Objective Function}
Since $\hat{\pazocal{J}}_N(K)$ is smooth with respect to $K$, we develop the second-order derivative of $\hat{\pazocal{J}}_N(K)$. To avoid using tensors, consider $\nabla^2 \hat{\pazocal{J}}_N(K)[E,E]$ as the action of the Hessian $\nabla^2 \hat{\pazocal{J}}_N(K)$ on $E \in \R^{n_u \times n_x}$. 

\begin{lem}\label{lm:Hess}
For all $K \in \hat{\pazocal{S}}_N$, the action of the Hessian $\nabla^2 \hat{\pazocal{J}}_N(K)$ on $E \in \R^{n_u \times n_x}$ is given by
\begin{align}\label{eq:Hess}
\begin{split}
    \nabla^2 \hat{\pazocal{J}}_N(K)[E,E] &= 2\innprod{I_{N+1} \otimes E}{(I_{N+1} \otimes RE){\pazocal{Y}}_{N}(K)} \\
    &-4\innprod{I_{N+1} \otimes E}{\pazocal{B}_N^\top{\pazocal{P}}'_N(K)[E]{\pazocal{Y}}_{N}(K)}
\end{split}
\end{align}
where ${\pazocal{P}}_N' = {\pazocal{P}}_N'(K)[E]$ is the action of the gradient ${\pazocal{P}}_N'(K)$ on $E \in \R^{m \times n}$, and is the solution of 
\begin{align}\label{eq:LyapuP'}
\begin{split}
     &\pazocal{A}_{N,c}(K)^\top {\pazocal{P}}_N' + {\pazocal{P}}_N' \pazocal{A}_{N,c}(K) \\
     & + (I_{N+1} \otimes E^\top)\pazocal{E}_N(K) + \pazocal{E}_N(K)^\top(I_{N+1} \otimes E) = 0.   
\end{split}
\end{align}
\end{lem}
\begin{proof}
It follows from \eqref{eq:gradientPCE1} that
\begin{align}\label{eq:Hess1}
\begin{split}
    &\nabla^2 {\pazocal{J}}_N(K)[E,E] \\
    &= 2\innprod{I_{N+1} \otimes E}{(I_{N+1} \otimes RE - \pazocal{B}_N^\top{\pazocal{P}}'_N){\pazocal{Y}}_{N}(K)} \\
    &+ 2\innprod{I_{N+1} \otimes E}{\pazocal{E}_N(K){\pazocal{Y}}_{N}'}
\end{split}
\end{align}
where ${\pazocal{Y}}_N' = {\pazocal{Y}}_N'(K)[E]$ is the action of the gradient ${\pazocal{Y}}_N'(K)$ on $E \in \R^{n_u \times n_x}$, and is the solution of 
\begin{align}\label{eq:LyapuY'}
\begin{split}
     &\pazocal{A}_{N,c}(K) {\pazocal{Y}}_N' + {\pazocal{Y}}_N' \pazocal{A}_{N,c}(K)^\top  - \pazocal{B}_N(I_{N+1} \otimes E){\pazocal{Y}}_N(K) \\
     &\, - {\pazocal{Y}}_N(K)(I_{N+1} \otimes E^\top)\pazocal{B}_N^\top  = 0. 
\end{split}
\end{align}
Applying Corollary \ref{cor:LyapPairtrace} to \eqref{eq:LyapuP'} and \eqref{eq:LyapuY'} yields
\begin{align}\label{eq:Hess2}
\begin{split}
    &2\innprod{I_{N+1} \otimes E}{\pazocal{E}_N(K){\pazocal{Y}}_{N}'} \\
    &= -2\Tr\!\left[ (I_{N+1} \otimes E^\top)\pazocal{B}_N^\top{\pazocal{P}}_N' {\pazocal{Y}}_N(K) \right]
\end{split}
\end{align}
Therefore, \eqref{eq:Hess} can be obtained by plugging \eqref{eq:Hess2} into \eqref{eq:Hess1}.
\end{proof}

\section{Theoretical Analysis of the Gradient Descent}
This section provides a rigorous convergence analysis of the gradient descent algorithm in \eqref{eq:gradientDescent} and quantifies the PCE approximation error $\hat{\pazocal{J}}_{N}(K)-\pazocal{J}(K)$ as $N$ increases.

Before the analysis, we introduce the notation used throughout. Fix a nominal parameter
$\bar{\xi}\in\Xi$, and define
\begin{align}\label{eq:parAbarBbarDef}
\begin{split}
    \bar{\pazocal{A}}_N
    &= \diag\big(A(\bar{\xi}),A(\bar{\xi}),\ldots,A(\bar{\xi})\big),\\
    \bar{\pazocal{B}}_N
    &= \diag\big(B(\bar{\xi}),B(\bar{\xi}),\ldots,B(\bar{\xi})\big),\\
    \bar{\pazocal{P}}_N(K)
    &= \diag\big(P(K,\bar{\xi}),P(K,\bar{\xi}),\ldots,P(K,\bar{\xi})\big),\\
    \bar{\pazocal{Y}}_N(K)
    &= \diag\big(Y(K,\bar{\xi}),0_{n_x\times n_x},\ldots,0_{n_x\times n_x}\big).
\end{split}
\end{align}
Additionally, to simplify notation, define the closed-loop state matrix for the nominal system as
\[
\bar{\pazocal{A}}_{N,c}(K):=\bar{\pazocal{A}}_N-\bar{\pazocal{B}}_N(I_{N+1}\otimes K).
\]
When there is no parametric uncertainty, i.e., $\Xi=\{\bar{\xi}\}$, we have
$\pazocal{A}_N=\bar{\pazocal{A}}_N$, $\pazocal{B}_N=\bar{\pazocal{B}}_N$,
$\pazocal{P}_N(K)=\bar{\pazocal{P}}_N(K)$, and
$\pazocal{Y}_N(K)=\bar{\pazocal{Y}}_N(K)$. By \eqref{eq:LyapuP} and \eqref{eq:LyapuY}, the nominal matrices satisfy
\begin{subequations}
\begin{align}
    &\bar{\pazocal{A}}_{N,c}^\top \bar{\pazocal{P}}_N(K) + \bar{\pazocal{P}}_N(K)\bar{\pazocal{A}}_{N,c}  + I_{N+1}\otimes\big(Q+K^\top R K\big)=0 \label{eq:AbarBbarLyapuP} \\
    &\bar{\pazocal{A}}_{N,c}\bar{\pazocal{Y}}_N(K) +\bar{\pazocal{Y}}_N(K)\bar{\pazocal{A}}_{N,c}^\top
    + \pazocal{I}_N \pazocal{I}_N^\top=0. \label{eq:AbarBbarLyapuY}
\end{align}    
\end{subequations}

For each parameter value $\xi$, define the set of stabilizing feedback gains as
\[
\pazocal{S}(\xi)
:=\big\{\,K\in\mathbb{R}^{n_u\times n_x} \ | \ A(\xi)-B(\xi)K \text{ is Hurwitz}\,\big\}.
\]
For a prescribed performance level $h>0$, define the sublevel set
\[
\pazocal{S}(\xi,h) :=\big\{\,K\in \pazocal{S}(\xi) \ | \ \pazocal{J}_{\xi}(K)\le h\,\big\}.
\]

The following assumption ensures that the uncertainty set $\Xi$ is sufficiently small.
\begin{assum}\label{ass:heterogeneity}
For any $\xi_1,\xi_2\in\Xi$,
\[
\|A(\xi_1)-A(\xi_2)\|\le \epsilon(h),
\qquad
\|B(\xi_1)-B(\xi_2)\|\le \epsilon(h),
\]
where $h>0$, $\epsilon_1 \in (0,1)$ and
\[
\epsilon(h):=\frac{\lambda_{\min}(Q)}{4h\big(1+a_1h+a_2\sqrt{h}\big)}
\min\!\left\{1,\frac{\epsilon_1}{h},\frac{\lambda_{\min}(Q)\epsilon_1}{h}\right\}.
\]
Here $a_1$ and $a_2$ are constants defined in Lemma~\ref{lm:boundKbyJK} with $(A,B)$ replaced by
$\big(A(\bar{\xi}),B(\bar{\xi})\big)$.
\end{assum}

Under Assumption~\ref{ass:heterogeneity}, together with \eqref{eq:LTIPCEABR} and Lemma~\ref{lm:normApceBound}, we obtain
\begin{align}\label{eq:ABerror}
\begin{split}
\|\tilde{\pazocal{A}}_N\|
&=\big\|\E\big[\phi_N(\xi)\phi_N(\xi)^\top \otimes \big(A(\xi)-A(\bar{\xi})\big)\big]\big\|
\le \epsilon\\
\|\tilde{\pazocal{B}}_N\|
&=\big\|\E\big[\phi_N(\xi)\phi_N(\xi)^\top \otimes \big(B(\xi)-B(\bar{\xi})\big)\big]\big\|
\le \epsilon
\end{split}
\end{align}
where 
\begin{align}\label{eq:parAtildeBtildeDef}
    \tilde{\pazocal{A}}_N = {\pazocal{A}}_N - \bar{\pazocal{A}}_N, \quad \tilde{\pazocal{B}}_N = {\pazocal{B}}_N - \bar{\pazocal{B}}_N .
\end{align}

\subsection{Convergence Rate of Gradient Descent}
In this subsection, we analyze the convergence rate of the gradient descent algorithm. To streamline the notation, define $S(K)\in\mathbb{S}_{++}^{n_x}$ as the unique solution to
\begin{align}\label{eq:SKdef}
    [\pazocal{A}_{N,c}(K)]_{0,0}\,S(K) + S(K)\,[\pazocal{A}_{N,c}(K)]_{0,0}^\top + I_{n_x} = 0.
\end{align}
In addition, define
\begin{align}\label{eq:defpazocalE}
\begin{split} 
    \bar{\pazocal{E}}_N(K) = I_{N+1} \otimes RK - \bar{\pazocal{B}}_N^\top\bar{\pazocal{P}}_N(K)
\end{split}    
\end{align}
We first show that $\pazocal{S}(\bar{\xi}, h)$ is forward invariant for \eqref{eq:gradientDescent} when $\epsilon_1$ in Assumption \ref{ass:heterogeneity} is sufficiently small.

\begin{lem}[Forward invariance of $\pazocal{S}(\bar{\xi},h)$]\label{lm:invariantS}
Let $\{K_k\}_{k\ge 0}$ be generated by the gradient descent iteration \eqref{eq:gradientDescent}. Assume that the step sizes satisfy $\eta_k \le 1/L_2(h)$ for all $k\ge 0$, where $L_2(h)$ is the Lipschitz constant of $\nabla \pazocal{J}_{\bar{\xi}}$ on $\pazocal{S}(\bar{\xi},h)$ given in Lemma~\ref{lm:Lsmoothness}. Under Assumption~\ref{ass:heterogeneity}, choose $\epsilon_1 \in (0,1)$ such that
\begin{equation}\label{eq:epsilon4}
\frac{2a_4\,n_x\,b_1(h)\,\epsilon_1}{1-2a_3\,n_x\,b_1(h)\,\epsilon_1}
\le h-\pazocal{J}_{\bar{\xi}}(K_{\bar{\xi}}^*),
\end{equation}
where $a_3$ and $a_4$ are constants defined in Lemma~\ref{lm:KPLcondition}. Then $\pazocal{S}(\bar{\xi},h)$ is forward invariant for \eqref{eq:gradientDescent}; that is, if $K_0\in \pazocal{S}(\bar{\xi},h)$, then $K_k\in \pazocal{S}(\bar{\xi},h)$ for all $k\ge 0$.
\end{lem}

\begin{proof}
    The gradient descent algorithm in \eqref{eq:gradientDescent} can be rewritten as 
    \begin{align}
        K_{k+1} = K_{k} -\eta_k  \nabla {\pazocal{J}}_{\bar \xi}(K_k)  + \eta_k  \rho(K_{k}),
    \end{align}
    where $\rho(K_{k}) = \nabla {\pazocal{J}}_{\bar \xi}(K_k) - \nabla \hat{\pazocal{J}}_N(K_k)$, and $\normF{\rho(K_{k})} \le  n_x b_1(h)\epsilon_1$ for $K_k \in \pazocal{S}(\bar \xi,h)$ by Lemma \ref{lm:gradientHet}. Suppose $K_k \in \pazocal{S}(\bar \xi,h)$ and denote 
    \begin{align}
        \kappa(k,s) = \pazocal{J}_{\bar{\xi}}\big(K_k - s(\nabla {\pazocal{J}}_{\bar \xi}(K_k) - \rho(K_{k}))\big).
    \end{align}
    The derivative of $\kappa(k,s)$ is written as
    \begin{align}\label{eq:kappaDerivative}
    \begin{split}
        &\frac{\partial \kappa(k,s)}{\partial s}\bigg\lvert_{s=0} = -\innprod{\nabla \pazocal{J}_{\bar{\xi}}\left(K_k\right)}{\nabla \pazocal{J}_{\bar{\xi}}(K_k) - \rho(K_k)}.
    \end{split}
    \end{align}

    Since $\pazocal{S}(\bar \xi, h)$ is compact \cite[Corollary 3.3.1.]{bu2020policy}, $K_k - s(\nabla {\pazocal{J}}_{\bar \xi}(K_k)  - \rho(K_{k}))$ must reach the boundary of $\pazocal{S}(\bar \xi, h)$ for some step size $\bar{s}_k > 0$. Let $\bar{s}_k$ be the first point of reaching the boundary of $\pazocal{S}(\bar \xi, h)$, that is $\pazocal{J}_{\bar{\xi}}(K_k - s(\nabla {\pazocal{J}}_{\bar \xi}(K_k) - \rho(K_{k}))) \le h$ for all $s \le \bar{s}_k$ and $\pazocal{J}_{\bar{\xi}}(K_k - \bar{s}_k(\nabla {\pazocal{J}}_{\bar \xi}(K_k)  - \rho(K_{k}))) = h$. Since $\kappa(k,s)$ is $L_3(h)$-smooth over $[0,\bar s_k]$ with $L_3(h) = \normF{\nabla \pazocal{J}_{\bar{\xi}}(K_k) - \rho(K_k)}^2L_2(h)$, where $L_2$ is defined in Lemma \ref{lm:Lsmoothness}, one can apply the smoothness condition to control the evolution of $\kappa$. According to \cite[Lemma 1.2.3]{nesterov2013introductory}, it holds that
    \begin{align}\label{eq:kappaDerivative1}
    \begin{split}
       &\kappa(k,\bar{s}_k) \le \kappa(k,0) + \frac{\partial \kappa(k,s)}{\partial s}\bigg\lvert_{s=0}\bar{s}_k  \\
       &\qquad + \frac{L_2(h)}{2}\normF{\nabla \pazocal{J}_{\bar \xi}(K_k)-\rho(K_k)}^2\bar{s}_k^2 \\
       &\quad = \kappa(k,0) - (\bar{s}_k - \frac{L_2(h)}{2}\bar{s}_k^2) \normF{\nabla \pazocal{J}_{\bar \xi}(K_k)}^2 \\
       &\qquad + (\bar{s}_k - L_2(h)\bar{s}_k^2) \innprod{\nabla \pazocal{J}_{\bar \xi}(K_k)}{\rho(K_k)} \\
       &\qquad + \frac{L_2(h)}{2}\bar{s}_k^2\normF{\rho(K_k)}^2
    \end{split}
    \end{align}

    We show that $\bar{s}_k \ge {1}/{L_2(h)}$ by contradiction. Suppose $\bar{s}_k < {1}/{L_2(h)}$, which implies that $(\bar{s}_k - L_2(h)\bar{s}_k^2) > 0$. Applying Young's inequality and the $\pazocal{K}$-PL condition in Lemma \ref{lm:KPLcondition} gives
    \begin{align}
    \begin{split}
        \kappa(k,\bar{s}_k) &\le \kappa(k,0) - \frac{\bar s_k}{2} \alpha(\pazocal{J}_{\bar \xi}(K_k)- \pazocal{J}_{\bar \xi}(K^*_{\bar \xi}))^2 \\
        &\quad + \frac{\bar s_k}{2}n_x^2b_1(h)^2\epsilon_1^2
    \end{split}
    \end{align}
    We consider two cases:
    
    \textit{Case 1}: 
    If $\alpha(\pazocal{J}_{\bar \xi}(K_k) - \pazocal{J}_{\bar \xi}(K_{\bar \xi}^*)) \ge 2n_x b_1(h) \epsilon_1$, the inequality yields $h = \kappa(k,\bar{s}_k) < h$, which leads to a contradiction.
    
    \textit{Case 2}:
    If $\alpha(\pazocal{J}_{\bar \xi}(K_k) - \pazocal{J}_{\bar \xi}(K_{\bar \xi}^*)) < 2n_x b_1(h)\epsilon_1$, then it follows from \eqref{eq:epsilon4} and Lemma \ref{lm:KPLcondition} that
    \begin{align}
    \begin{split}
        h &< \alpha^{-1}(2 n_x b_1(h) \epsilon_1) + \pazocal{J}_{\bar \xi}(K_{\bar \xi}^*) \\
        &= \frac{2a_4 n_x b_1(h) \epsilon_1}{1-2a_3 n_x b_1(h) \epsilon_1} + \pazocal{J}_{\bar \xi}(K_{\bar \xi}^*) \le h
    \end{split}
    \end{align}
    which leads to another contradiction. In both cases, we reach a contradiction. Therefore, it must hold that $\bar{s}_k \ge {1}/{L_2(h)}$. 

    Since $K_k \in \pazocal{S}(\bar \xi, h)$, and $\bar s_k$  is defined as the first time the trajectory reaches the boundary of $\pazocal{S}(\bar \xi, h)$, it follows that any step size $\eta_k \le \frac{1}{L_2(h)} \le \bar s_k$ keeps the next iterate within the set. Therefore, $K_{k+1} \in \pazocal{S}(\bar \xi, h)$. Given that the initial condition $K_0 \in \pazocal{S}(\bar \xi, h)$, we conclude by induction that the set $ \pazocal{S}(\bar \xi, h)$ is forward invariant.
    
\end{proof}

The following lemma shows that $\hat{\pazocal{J}}_N$ satisfies a Polyak--\L{}ojasiewicz (PL) inequality perturbed by the heterogeneity specified in Assumption~\ref{ass:heterogeneity}.

\begin{lem}[Perturbed PL inequality]\label{lm:PL-het}
Under Assumption~\ref{ass:heterogeneity} and suppose that a global minimizer $K^*$ of $\hat{\pazocal{J}}_N$ satisfies $K^*\in \pazocal{S}(\bar{\xi},h)$. Then, for every $K\in \pazocal{S}(\bar{\xi},h)$, the perturbed PL inequality
\begin{equation}\label{eq:PL-het}
\|\nabla \hat{\pazocal{J}}_N(K)\|_{\mathrm F}^2
\ge
b_4(h)\big(\hat{\pazocal{J}}_N(K)-\hat{\pazocal{J}}_N(K^*)\big)
- b_5(h)\,\epsilon_1
\end{equation}
holds, where
\begin{align*}
b_4(h)&:=\frac{b_2(h)^2\lambda_{\min}(R)}{2(h+1)},\\
b_5(h)&:=b_3(h)^2 + \frac{\,n_x\,\lambda_{\min}(Q)\lambda_{\min}(R)\,b_2(h)^2}{2h^2(h+1)},
\end{align*}
and the functions $b_2(h)$ and $b_3(h)$ are defined in Lemma~\ref{lm:gradLowerBound}.
\end{lem}

\begin{proof}
Omit the argument $K$ when there is no ambiguity. Let
\[
\Delta P:=\pazocal{P}_N(K)-\pazocal{P}_N(K^*),\qquad \Delta K:=K-K^*.
\]
Subtracting the Lyapunov equation \eqref{eq:LyapunovPPCE} at $K^*$ from that at $K$ results in
\begin{align}\label{eq:DeltaP-Lyap}
\begin{split}
0&=\pazocal{A}_{N,c}(K^*)^\top \Delta P +\Delta P\,\pazocal{A}_{N,c}(K^*) \\
& +(I_{N+1}\otimes \Delta K)^\top \pazocal{E}_N(K) +\pazocal{E}_N(K)^\top (I_{N+1}\otimes \Delta K) \\
& -(I_{N+1}\otimes \Delta K)^\top (I_{N+1}\otimes R)(I_{N+1}\otimes \Delta K).
\end{split}
\end{align}
Taking the $(0,0)$ block of \eqref{eq:DeltaP-Lyap} gives
\begin{align}\label{eq:DeltaP00}
\begin{split}
0 &=[\pazocal{A}_{N,c}(K^*)]_{0,0}^\top[\Delta P]_{0,0} +[\Delta P]_{0,0}[\pazocal{A}_{N,c}(K^*)]_{0,0} \\
& +\Delta K^\top[\pazocal{E}_N(K)]_{0,0}+[\pazocal{E}_N(K)]_{0,0}^\top \Delta K -\Delta K^\top R\Delta K \\
& +[\pazocal{A}_{N,c}(K^*)]_{1\!:\!N,0}^\top[\Delta P]_{1\!:\!N,0} +[\Delta P]_{0,1\!:\!N}[\pazocal{A}_{N,c}(K^*)]_{1\!:\!N,0}.
\end{split}
\end{align}

\smallskip
\noindent\emph{Step 1: Dominate the $(0,0)$ block by a comparison Lyapunov equation.}
Apply Lemma~\ref{lm:completeSquare} to the term
$\Delta K^\top[\pazocal{E}_N]_{0,0}+[\pazocal{E}_N]_{0,0}^\top\Delta K$:
\begin{align*}
\Delta K^\top[\pazocal{E}_N]_{0,0}+[\pazocal{E}_N]_{0,0}^\top\Delta K&-\Delta K^\top R\Delta K \\
&\quad \preceq \lambda_{\min}(R)^{-1}[\pazocal{E}_N]_{0,0}^\top[\pazocal{E}_N]_{0,0}.  
\end{align*}
Moreover, by Assumption~\ref{ass:heterogeneity}, and $[\bar{\pazocal{A}}_{N,c}(K^*)]_{1\!:\!N,0} = 0$,
\begin{equation}\label{eq:A10-bound}
\|[\pazocal{A}_{N,c}(K^*)]_{1\!:\!N,0}\| \le (1+\|K\|)\epsilon(h),
\end{equation}
and
\begin{align}\label{eq:DeltaPoff-bound}
\begin{split}
&\|[\Delta P]_{1\!:\!N,0}\|+\|[\Delta P]_{0,1\!:\!N}\|
\le 2\|\Delta P\| \\
&\quad \le 2(\|\pazocal{P}_N(K)\|+\|\pazocal{P}_N(K^*)\|)
\le 2(h+\epsilon_1),    
\end{split}
\end{align}
where the last inequality uses Lemma \ref{lm:PYLyapuheter}.

Combining \eqref{eq:DeltaP00}--\eqref{eq:DeltaPoff-bound} gives the matrix inequality
\begin{align}\label{eq:DeltaP00-ineq}
\begin{split}
&[\pazocal{A}_{N,c}(K^*)]_{0,0}^\top[\Delta P]_{0,0} +[\Delta P]_{0,0}[\pazocal{A}_{N,c}(K^*)]_{0,0} \\
&+\lambda_{\min}(R)^{-1}[\pazocal{E}_N]_{0,0}^\top[\pazocal{E}_N]_{0,0} + \Gamma \preceq 0,
\end{split}
\end{align}
where
\begin{equation}\label{eq:GammaDef}
\Gamma :=2\,\epsilon(h)\,(1+\|K\|)\,(h+\epsilon_1)\,I_{n_x}.
\end{equation}
By Lemma \ref{lm:boundSK}, $[\pazocal{A}_{N,c}(K^*)]_{0,0}$ is Hurwitz.  Corollary~\ref{cor:LyapCompare} and \eqref{eq:DeltaP00-ineq} implies
$[\Delta P]_{0,0}\preceq D$, where $D\succeq 0$ is the unique solution of
\begin{align}\label{eq:D-Lyap}
\begin{split}
&[\pazocal{A}_{N,c}(K^*)]_{0,0}^\top D + D[\pazocal{A}_{N,c}(K^*)]_{0,0}\\
&\quad +\lambda_{\min}(R)^{-1}[\pazocal{E}_N]_{0,0}^\top[\pazocal{E}_N]_{0,0} +\Gamma=0.
\end{split}
\end{align}

\smallskip
\noindent\emph{Step 2: Relate $\hat{\pazocal{J}}_N(K)-\hat{\pazocal{J}}_N(K^*)$ to $D$.}
By definition in \eqref{eq:surrogateOpt},
\[
\hat{\pazocal{J}}_N(K)-\hat{\pazocal{J}}_N(K^*)
=\Tr\big([\Delta P]_{0,0}\big)
\le \Tr(D).
\]
Using Corollary~\ref{cor:LyapPairtrace} on \eqref{eq:D-Lyap} gives
\begin{align}\label{eq:TrD-bound}
\begin{split}
&\Tr(D)=\tfrac{1}{\lambda_{\min}(R)}\Tr\!\big([\pazocal{E}_N]_{0,0}^\top[\pazocal{E}_N]_{0,0}\,S(K^*)\big) +\Tr\!\big(\Gamma\,S(K^*)\big) \\
&\le \frac{\|S(K^*)\|}{\lambda_{\min}(R)}\,\|[\pazocal{E}_N]_{0,0}\|_{\mathrm F}^2 + \norm{\Gamma}\,\Tr(S(K^*)) \\
&\le \frac{h+\epsilon_1}{\lambda_{\min}(R)}\,\|[\pazocal{E}_N]_{0,0}\|_{\mathrm F}^2 \\
&\quad + 2n_x\,\epsilon(h)\,(1+a_1 + a_2\sqrt{h})\,(h+\epsilon_1)^2,
\end{split}
\end{align}
where we used Lemma \ref{lm:boundKbyJK} and $\|S(K^*)\|\le h+\epsilon_1$ in Lemma \ref{lm:boundSK}. Plugging Lemma~\ref{lm:gradLowerBound} together with Assumption~\ref{ass:heterogeneity} into \eqref{eq:TrD-bound} yields \eqref{eq:PL-het}.

\end{proof}

The following theorem establishes a linear (geometric) convergence rate for the gradient descent iteration \eqref{eq:gradientDescent}, up to a residual term induced by the heterogeneity level $\epsilon_1$.

\begin{thm}\label{thm:gradientConv}
Suppose the step sizes satisfy
\[
0 < \theta<\eta_k\le \min\Big\{\frac{1}{L_1(h)},\,\frac{1}{L_2(h)}\Big\},\qquad \forall k\ge 0,
\]
where $L_1(h)$ and $L_2(h)$ are the Lipschitz constants in Lemmas~\ref{lm:Lsmooth} and \ref{lm:Lsmoothness}, respectively. Assume the conditions of Lemmas~\ref{lm:invariantS} and~\ref{lm:PL-het} hold. Then, for all $k\ge 0$,
\begin{align}\label{eq:linear-rate}
\begin{split}
&\hat{\pazocal{J}}_N(K_k)-\hat{\pazocal{J}}_N(K^*) \le \\
&\Big(1-\frac{\theta b_4(h)}{2}\Big)^{\!k} \big(\hat{\pazocal{J}}_N(K_0)-\hat{\pazocal{J}}_N(K^*)\big) +{b_5(h)}\epsilon_1/{b_4(h)}.
\end{split}
\end{align}
\end{thm}
\begin{proof}
By Lemma \ref{lm:invariantS}, $K_k \in \pazocal{S}(\bar \xi, h)$ for all $k \ge 0$. The $L_1$-smoothness of $\hat{\pazocal{J}}_N(K)$ in Lemma \ref{lm:Lsmooth} and \cite[Lemma 1.2.3]{nesterov2013introductory} imply that
\begin{align}\label{eq:gradientconv1}
\begin{split}
    &\hat{\pazocal{J}}_N(K_{k+1}) - \hat{\pazocal{J}}_N(K^*) \le \hat{\pazocal{J}}_N(K_k) - \hat{\pazocal{J}}_N(K^*) \\
    &\quad - \big(1 - \tfrac{L_1(h)}{2}\eta_k\big)\eta_k \normF{\nabla \hat{\pazocal{J}}_N(K_k)}^2 \\
    &\le (1- \frac{\theta b_4(h)}{2}) ( \hat{\pazocal{J}}_N(K_k) - \hat{\pazocal{J}}_N(K^*)) + \frac{\theta b_5(h)}{2} \epsilon_1
\end{split}
\end{align}
where the last line is from Lemma \ref{lm:PL-het}. Hence the proof is completed.

\end{proof}

\subsection{Convergence of the PCE Approximation}

This subsection evaluates the PCE approximation accuracy in the scalar-parameter case $\Xi\subseteq\mathbb{R}$ with $\xi\sim U(\Xi)$, where the orthonormal basis
$\{\varphi_i\}_{i=0}^\infty$ consists of (normalized) Legendre polynomials.

For $p\in\mathbb{Z}_+$, define the (weighted) Sobolev space
\begin{equation}
\pazocal{F}_2^p :=\Big\{\psi:\Xi\to\mathbb{R}\ \big|\ \psi^{(i)}\in \pazocal{L}_2,\ i=0,1,\ldots,p\Big\},
\end{equation}
where
\begin{equation}
\pazocal{L}_2 :=\Big\{\psi:\Xi\to\mathbb{R}\ \big|\ \int_{\Xi}|\psi(\xi)|^2\,\de\xi<\infty\Big\}.
\end{equation}
Equip $\pazocal{F}_2^p$ with the inner product
\begin{align}
\begin{split}
\langle \psi_1,\psi_2\rangle_{\pazocal{F}_2^p} &:=\sum_{i=0}^p \left\langle \psi_1^{(i)},\psi_2^{(i)}\right\rangle_{\pazocal{L}_2}, \\
\langle f,g\rangle_{\pazocal{L}_2} &:=\int_{\Xi} f(\xi)g(\xi)\,\de\xi .
\end{split}
\end{align}

The following lemma states the spectral (algebraic) convergence rate of the Legendre PCE approximation.
\begin{lem}[{\cite[Thm.~3.6]{Xiu2010}}]\label{lm:xiu-thm36}
Let $p\ge 1$ and $\psi\in \pazocal{F}_2^p$. Then there exists a constant $C>0$,
independent of $N$, such that
\vspace{-0.2cm}
\begin{equation}
\big\|\psi-\big(\E[\psi(\xi)\phi_N(\xi)]\big)^\top \phi_N(\xi)\big\|_{\pazocal{L}_2}
\le C N^{-p}\,\|\psi\|_{\pazocal{F}_2^p}.
\end{equation}
\end{lem}

\vspace{0.2cm}

The below lemma shows that the PCE truncation error of $x(t,\xi)$ decays at rate $O(N^{-p})$.
\begin{lem}\label{lm:eNbound_time}
Under Assumption~\ref{ass:heterogeneity}, for any fixed $K\in\pazocal{S}(h,\bar{\xi})$ and any $t\ge 0$, the truncation error in \eqref{eq:xPCE},
\[
e_N(t,\xi)
:=x(t,\xi)-\E\big[x(t,\xi)\phi_N(\xi)\big]^\top\phi_N(\xi),
\]
satisfies
\begin{equation}\label{eq:eNbound_time}
\|e_N(t,\cdot)\|_{\pazocal{L}_2}
\le
CN^{-p}
\Bigg(\sum_{i=0}^{p}\sum_{j=0}^{i} d_{i,j}(h)\,t^{j}\Bigg) e^{-\gamma t}\|x_0\|,
\end{equation}
where $C>0$ is independent of $N$ (and $t$), and $d_{i,j}(h)$ is as in Lemma~\ref{lm:PhiDerivBound}.
\end{lem}

\begin{proof}
By definition of the $\pazocal{F}_2^p$-norm and $x(t,\xi)=\Phi(t,\xi)x_0$, with $\Phi(t,\xi) = e^{A_c(\xi)t}$, we have
\begin{align*}
\|x(t,\cdot)\|_{\pazocal{F}_2^p}^2 &=\sum_{i=0}^{p}\big\|\Phi^{(i)}(t,\cdot)\,x_0\big\|_{\pazocal{L}_2}^2 \\
&\le \sum_{i=0}^{p}\int_{\Xi}\|\Phi^{(i)}(t,\xi)\|^2\,\de\mathbb{P}(\xi)\,\|x_0\|^2 \\
&\le \sum_{i=0}^{p}\sup_{\xi\in\Xi}\|\Phi^{(i)}(t,\xi)\|^2\,\|x_0\|^2 .
\end{align*}
Applying Lemma~\ref{lm:PhiDerivBound} gives
\[
\|x(t,\cdot)\|_{\pazocal{F}_2^p}
\le
\Bigg(\sum_{i=0}^{p}\sum_{j=0}^{i} d_{i,j}(h)\,t^{j}e^{-\gamma t}\Bigg)\|x_0\|.
\]
Finally, applying the Legendre projection error bound (Lemma~\ref{lm:xiu-thm36}) to $\psi(\xi)=x(t,\xi)$ yields \eqref{eq:eNbound_time}.
\end{proof}

The surrogate LTI cost $\hat{\pazocal{J}}_N$ in \eqref{eq:surrogateJ} obtained via PCE converges to the true cost $\pazocal{J}$ in \eqref{eq:cost} at an algebraic rate $O(N^{-p})$.

\begin{thm}\label{thm:PCEApproConv}
Under Assumption~\ref{ass:heterogeneity}, for any $K\in\pazocal{S}(\bar{\xi},h)$ and any $p\in\mathbb{Z}_+$, there exists a constant $C_1(h)>0$, independent of $N$, such that
\[
\big|\pazocal{J}(K)-\hat{\pazocal{J}}_N(K)\big| \le C_1(h) N^{-p}.
\]
\end{thm}
\begin{proof}

\emph{Step 1: Bounding the forcing term.}
Recall the residual term $\pazocal{R}_N(t)$ in the lifted dynamics. By the Cauchy--Schwarz inequality (Lemma~\ref{lm:CSIneqMatrix}),
\begin{align*}
\|\pazocal{R}_N(t)\| &\le \Big\|\E \big[\phi_N(\xi)\phi_N(\xi)^\top\otimes A_c(\xi,K)A_c(\xi,K)^\top\big]\Big\|^{1/2} \\
&\quad \times \|e_N(t,\cdot)\|_{\pazocal{L}_2}.
\end{align*}
Since
\begin{align*}
    &\|A_c(\xi,K)\|\le \bar a+\bar b\|K\|, \, \bar a:=\sup_{\xi\in\Xi}\|A(\xi)\|, \, \bar b:=\sup_{\xi\in\Xi}\|B(\xi)\|,    
\end{align*}
Lemmas~\ref{lm:boundKbyJK} and \ref{lm:normApceBound} yield
\[
\|\pazocal{R}_N(t)\|
\le
(\bar a+\bar b(a_1h+a_2\sqrt h))\,\|e_N(t,\cdot)\|_{\pazocal{L}_2}.
\]
Applying Lemma~\ref{lm:eNbound_time},
\begin{align}\label{eq:RupperBound}
\|\pazocal{R}_N(t)\|
\le
C_R(h)\,N^{-p}\|x_0\|,
\end{align}
where $C_R(h)>0$ is independent of $N$ and $t$.

\emph{Step 2: Error dynamics of the lifted state.}
Let $\tilde z_N:=z_N-\hat z_N$. From \eqref{eq:zdot} and \eqref{eq:surrogateLTI},
\[
\dot{\tilde z}_N = (\pazocal{A}_N-\pazocal{B}_N(I_{N+1}\otimes K))\tilde z_N+\pazocal{R}_N(t), \quad \tilde z_N(0)=0.
\]
Using the Lyapunov function $\tilde z_N^\top \pazocal{P}_N(K)\tilde z_N$ and Lemma~\ref{lm:PYLyapuheter},
\begin{align*}
\frac{\de}{\de t}\big(\tilde z_N^\top \pazocal{P}_N\tilde z_N\big)
&\le
-\lambda_{\min}(Q)\|\tilde z_N\|^2 +2(h+\epsilon_1)\|\pazocal{R}_N(t)\|\,\|\tilde z_N\| \\
&\le 
-\tfrac{\lambda_{\min}(Q)}{2}\|\tilde z_N\|^2 +\tfrac{2(h+\epsilon_1)^2}{\lambda_{\min}(Q)}\|\pazocal{R}_N(t)\|^2.
\end{align*}
Since $\tilde z_N(0)=0$, it follows from \eqref{eq:RupperBound} that
\[
\|\tilde z_N(t)\|
\le
C_z(h) N^{-p}\|x_0\|,
\]
for some $C_z(h)>0$.

\emph{Step 3: State approximation error.}
Let $\hat x(t,\xi)=(\phi_N(\xi)^\top\otimes I_{n_x})\hat z_N(t)$ and $\tilde x=x-\hat x$.
Then
\[
\|\tilde x(t,\cdot)\|_{\pazocal{L}_2}
\le
\|e_N(t,\cdot)\|_{\pazocal{L}_2}+\|\tilde z_N(t)\|
\le
C_x(h) N^{-p}\|x_0\|.
\]

\emph{Step 4: Bounding the cost difference.}
Using the quadratic cost structure,
\begin{align*}
&|\pazocal{J}(K)-\hat{\pazocal{J}}_N(K)| \le (\|Q\|+\|R\|\|K\|^2) \\
&\quad \times  \int_0^\infty \|\tilde x(t,\cdot)\|_{\pazocal{L}_2}
\big(\|x(t,\cdot)\|_{\pazocal{L}_2}
+\|\hat x(t,\cdot)\|_{\pazocal{L}_2}\big)\,\de t.
\end{align*}
By Lemmas~\ref{lm:ExpDecay} and \ref{lm:PhiBound}, both $x(t,\xi)$ and $\hat x(t,\xi)$ decay exponentially. Hence the integral is finite and
\[
|\pazocal{J}(K)-\hat{\pazocal{J}}_N(K)|
\le
C_1(h) N^{-p},
\]
for some constant $C_1(h)>0$ independent of $N$.
\end{proof}

\begin{rem}
In \cite{fujinami2025policygradientlqrdomain}, a domain randomization (Monte Carlo) approach is used to approximate the cost in \eqref{eq:cost}, for which the approximation error decays at the rate $O(1/\sqrt{M})$, where $M$ denotes the number of samples. In contrast, the PCE-based approximation established in Theorem~\ref{thm:PCEApproConv} achieves an algebraic convergence rate of $O(N^{-p})$. Consequently, for sufficiently smooth parameter dependence, the PCE approach is more sample-efficient than Monte Carlo-based domain randomization.
\end{rem}

\section{Numerical Studies}
This section demonstrates the effectiveness of the proposed PO algorithm by an illustrative example and a mass-spring system.
\subsection{Illustrative Example}
We apply the proposed PO algorithm in \eqref{eq:gradientDescent} to the system studied in \cite{wan2021polynomial}, where
\begin{align*}
    A(\xi) = \begin{bmatrix}
        0.2+0.3\xi^3 &-0.4 \\
        0.1 &0.5
    \end{bmatrix}, \quad B(\xi) = \begin{bmatrix}
        0.5 &0.1 \\
        0.2 &1
    \end{bmatrix}
\end{align*}
and the probabilistic parameter $\xi$ is a uniformly distributed random variable over $[-1,1]$. According to the Askey scheme \cite{Xiu2010}, the orthonormal polynomials $\phi_N(\xi)$ are chosen as the Legendre polynomials. The weighting matrices of the LQR cost are $Q = R = I_2$. The initial control gain $K_0$ is designed by \cite[Theorem 4.3 (iii-c)]{Ebihara2015s} to start the gradient-based method. The step size of the gradient descent algorithm is $\eta_k = 10^{-2}$. The algorithm is stopped when $\norm{\nabla \hat{\pazocal{J}}_N(K)} \le 10^{-3}$.  

For $N=5$, Fig.\ \ref{fig:Jconv} presents the convergence results of the cost function and gradient during the optimization process. We observe a consistent decrease in the cost and gradient, indicating successful convergence of the optimization algorithm. The theoretical convergence rate established in Theorem~\ref{thm:gradientConv} is linear, which is consistent with the trend observed in the numerical results. The improved convergence performance compared to the theoretical rate indicates the efficiency of the algorithm in practice.

\begin{figure}[t]
    \begin{subfigure}[b]{.49\linewidth}
        \centering
        \includegraphics[width=\linewidth]{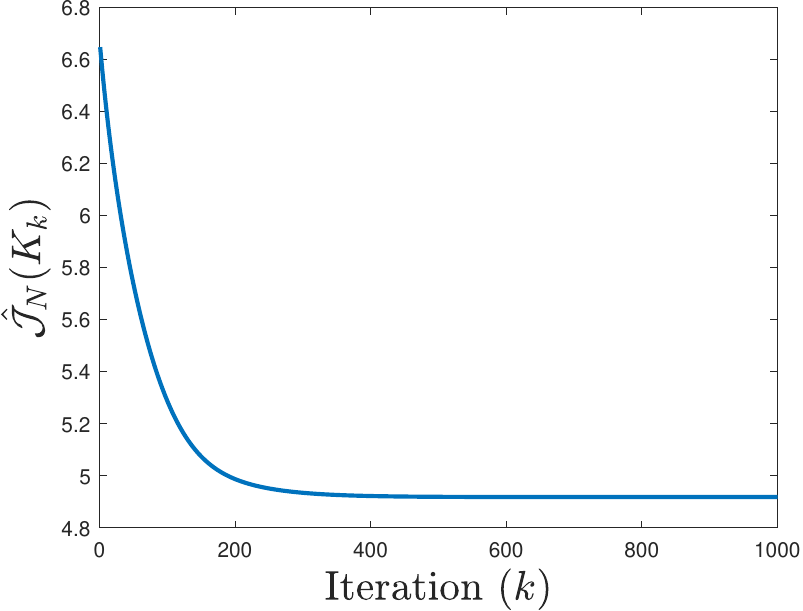}
    \end{subfigure}%
    \begin{subfigure}[b]{.49\linewidth}
        \centering
        \includegraphics[width=\linewidth]{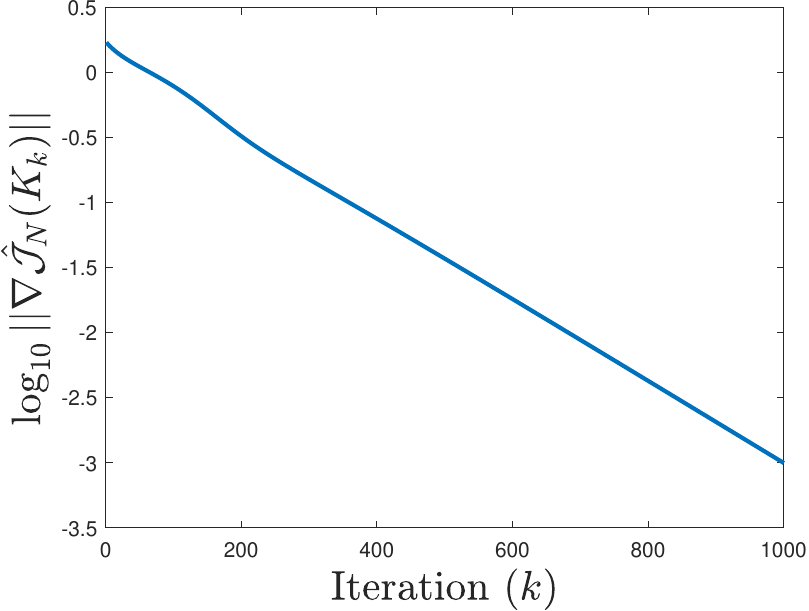}    
    \end{subfigure}%

\vspace{-0.2cm}
    
    \caption{For the illustrative example, the evolution of the surrogate cost $\hat{\pazocal{J}}_N(K_k)$ and gradient $\nabla \hat{\pazocal{J}}_N(K_k)$ for fifth-order orthonormal polynomials.}
    \label{fig:Jconv}
\end{figure}

The PO algorithm finally converges to $K_p = \begin{bmatrix}
1.25 &-0.10 \\
-0.82 &1.97    
\end{bmatrix}$. Figure \ref{fig:costCom} compares the cost function $\pazocal{J}_\xi(x_0,K)$ as defined in \eqref{eq:cost} for the controller $K_p$ generated by the PO algorithm in \eqref{eq:gradientDescent} to the controller $K_r$ designed by the S-variable approach in \cite[Theorem (iii-c)]{Ebihara2015s}. The comparison is carried out under various randomly sampled initial conditions $x_0$, with the results plotted against the parameter $\xi$. The plots illustrate that the controller optimized by the PO algorithm consistently achieves better performance than the S-variable approach, as indicated by the lower cost values for various initial conditions.

\begin{figure}[t]
    \centering
    \begin{subfigure}[b]{0.5\linewidth}
        \centering
        \includegraphics[width=\linewidth]{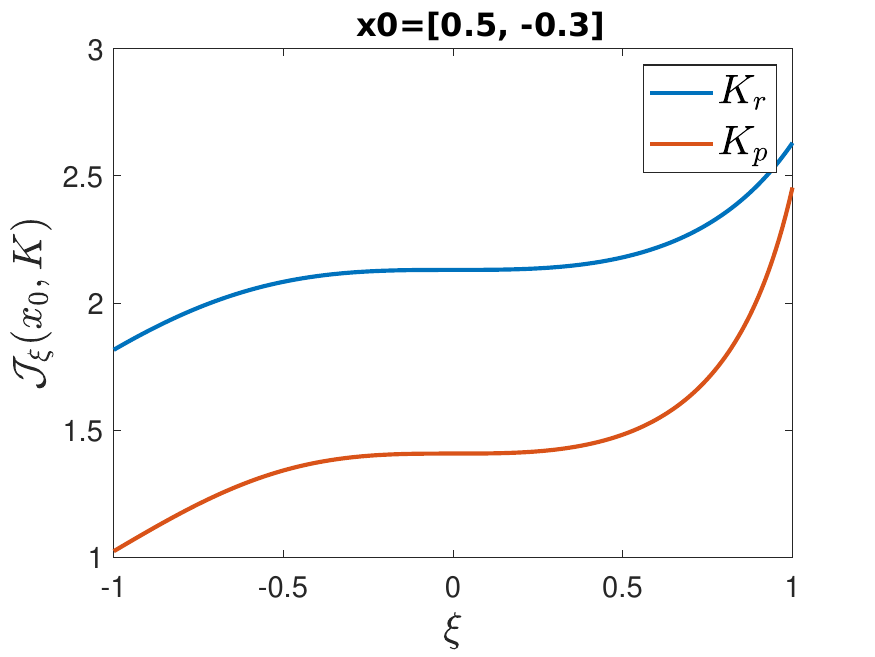}
    \end{subfigure}%
    \begin{subfigure}[b]{0.5\linewidth}
        \centering
        \includegraphics[width=\linewidth]{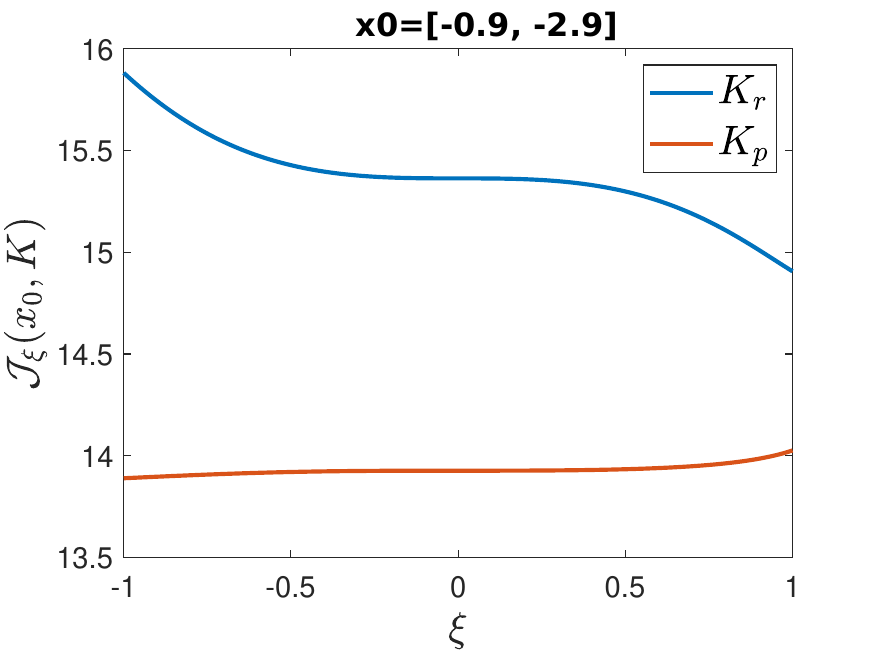}    
    \end{subfigure}%
    \\
    \begin{subfigure}[b]{0.5\linewidth}
        \centering
        \includegraphics[width=\linewidth]{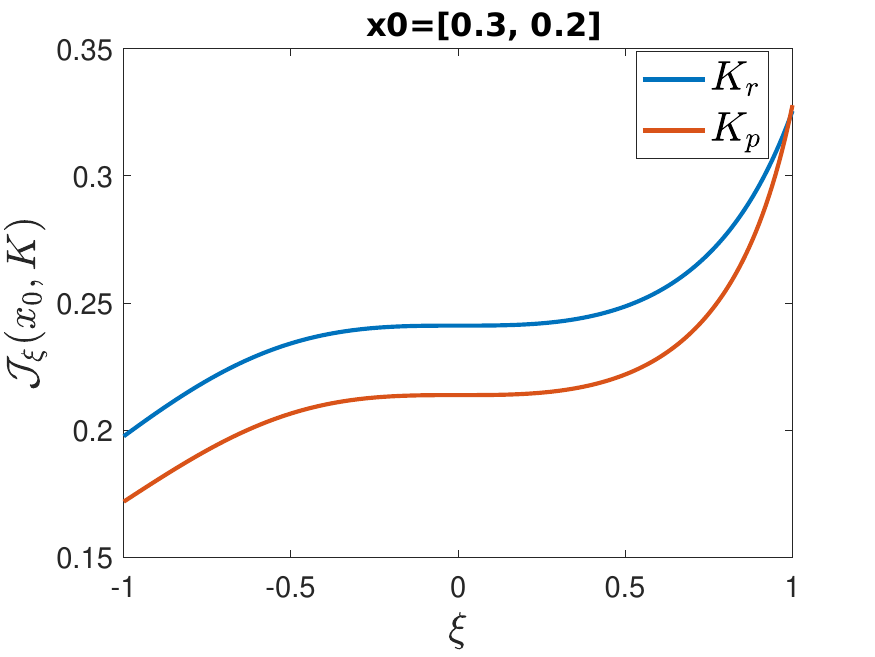}%
    \end{subfigure}%
    \begin{subfigure}[b]{0.5\linewidth}
        \centering
        \includegraphics[width=\linewidth]{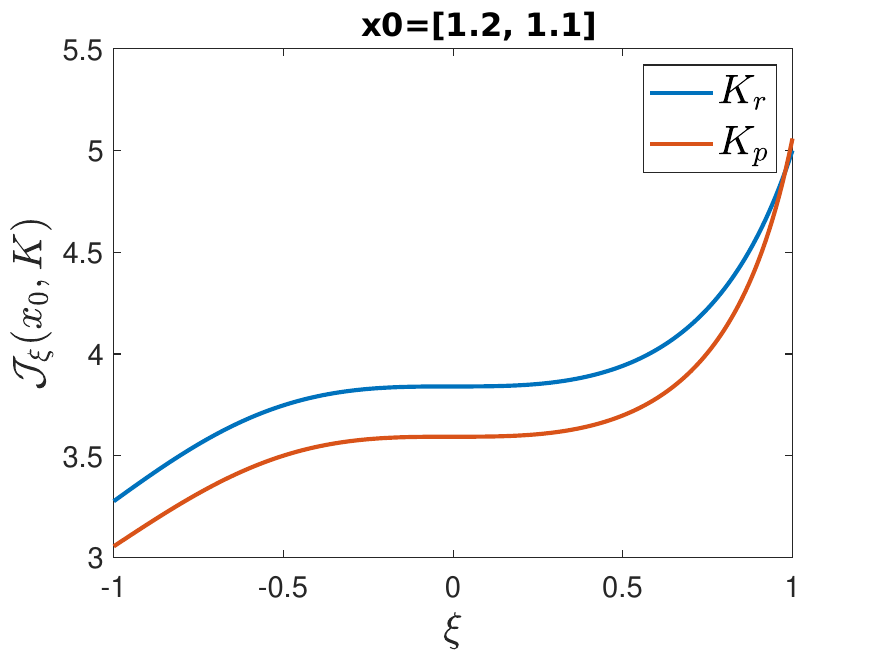}%
    \end{subfigure}%
    
\vspace{-0.2cm}
    
    \caption{For the illustrative example, the cost for the controller $K_r$ generated by the S-variable approach \cite[Theorem (iii-c)]{Ebihara2015s} and the optimized controller $K_p$ by the PO algorithm in \eqref{eq:gradientDescent}.}
    \label{fig:costCom}
\end{figure}

Table \ref{tab:Time} compares the computational time with the BMI-based approach, where the objective function is \eqref{eq:surrogateOpt} and the constraint is \eqref{eq:LyapunovPPCE} with ``$=$'' replaced by ``$\prec$''. The BMI-based optimization problem is solved using different methods, including PENLAB \cite{Fiala2013penlab}, the path-following (PF) approach \cite{Hassibi1999}, and the convex–concave decomposition (CCD) approach \cite{Tran2012}. Both the PF and CCD methods are initialized with the same starting controller $K_0$ as the proposed PO approach, ensuring a fair comparison in terms of performance and computational efficiency. All simulations are implemented $20$ times in Matlab R2024b on a personal computer with a 13th Gen Intel(R) Core(TM) i7-13700K CPU and a 32GB RAM. When the algorithms converge, the corresponding computational time and the cost are shown in Tables \ref{tab:Time} and \ref{tab:ComCost}, respectively. The PO algorithm is significantly faster than the other BMI-based methods across all problem sizes, especially when dealing with larger scale problems. Despite the differences in computational time, the control performance (as measured by the cost $\hat{\pazocal{J}}_N(K)$) is nearly the same for all methods. The PO algorithm achieved comparable performance to the BMI-based methods while being much faster.

\begin{table*}
    \caption{Comparison of computational times for the illustrative example (unit: second, Alg.\ = Algorithm)}
    \centering
    \adjustbox{max width=\textwidth}{%
    \begin{tabular}{|c||c c|c c|c c|c c|}
        \hline
        \multirow{2}{6.43em}{\diagbox{Order}{Alg.}} & \multicolumn{2}{c|}{PO} & \multicolumn{2}{c|}{PENLAB} & \multicolumn{2}{c|}{PF} & \multicolumn{2}{c|}{CCD} \\ \cline{2-9}
                               & Mean & Std & Mean & Std & Mean & Std & Mean & Std \\ \hline \hline
       $N=3$ & 0.33  & 0.14  & 0.44  & 0.12  & 2.43  & 0.09  & 1.19  & 0.02 \\ \hline
        $N=5$ & 0.39  & 0.14  & 1.48  & 0.14  & 2.63  & 0.09  & 1.60  & 0.02 \\ \hline
        $N=8$ & 0.53  & 0.15  & 7.24  & 0.51  & 3.28  & 0.09  & 3.14  & 0.03 \\ \hline
    \end{tabular}
    }
    \label{tab:Time}
\end{table*}

\begin{table}[thb]
    \caption{Comparison of the cost $\hat{\pazocal{J}}_N(K)$ for the illustrative example (Alg.\ = Algorithm)}
    \centering
    \begin{tabular}{|c||c|c|c|c|}
        \hline
        \diagbox{Order}{Alg.} & {PO} & {PENLAB} &{PF} & {CCD} \\ \hline \hline
        $N=3$ & 4.92   & 5.07 & 5.08  & 4.92 \\ \hline
        $N=5$ & 4.92  &  5.07  & 5.11 &  4.92  \\ \hline
        $N=8$ & 4.92   & 5.07  & 5.10  & 4.92 \\ \hline
    \end{tabular}
    \label{tab:ComCost}
\end{table}

\subsection{Mass-Spring System}

The two-mass-spring system in \cite{wie1992robust,braatz1992robust} is extended to the four-mass-spring system in Fig. \ref{fig:MassSpring}, whose system matrices are 
\begin{align*}
    A(\xi) &= \begin{bmatrix}
        0_{4 \times 4} &I_4 \\
        A_{21}(\xi) & 0_{4 \times 4}
        \end{bmatrix}, \\
    A_{21}(\xi) &= \begin{bmatrix}
    -\frac{\kappa(\xi)}{m_1}  &\frac{\kappa(\xi)}{m_1} &0  &0 \\
    \frac{\kappa(\xi)}{m_2}  &-2\frac{\kappa(\xi)}{m_2} &\frac{\kappa(\xi)}{m_2}  &0 \\
    0  &\frac{\kappa(\xi)}{m_3} &-2\frac{\kappa(\xi)}{m_3}  &\frac{\kappa(\xi)}{m_3} \\
    0  &0 &\frac{\kappa(\xi)}{m_4}  &-\frac{\kappa(\xi)}{m_4}
    \end{bmatrix}, \\
    B(\xi) &= \begin{bmatrix}
        0 &0 &0 &0 &\frac{1}{m_1} &0 &0  &0
    \end{bmatrix}^\top,
\end{align*}
where $m_1 = m_2 = m_3 = m_4 = 1$ are the masses, $\kappa(\xi) = (\frac{\xi}{5}+1)^4$ is the stiffness of the springs, and $\xi$ is a uniformly distributed random variable over $[-1,1]$. The cost matrices are $Q=I_8$ and $R = 1$. The step size of the PO algorithm is $\eta_k = {0.01}$. The algorithm stops when $\norm{\nabla \hat{\pazocal{J}}_N(K)} \le 10^{-3}$. 

\begin{figure}[thb]
\centering
\includegraphics[width=1\linewidth]{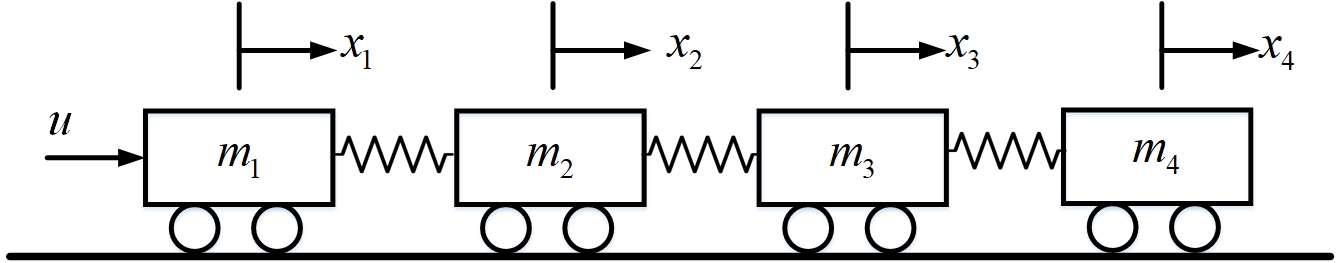}

\vspace{-0.2cm}
    
\caption{Four-mass-spring system.}
\label{fig:MassSpring}
\end{figure}

For $N=5$, the cost evolution during the optimization process is shown in Fig.\ \ref{fig:Jconv_pra}. The algorithm stops after about $1600$ iterations, and the gradient approaches to zero at a linear convergence rate. For $N=5$, the feedback gain finally converges to 
\begin{align*}
    K_p = [
        2.55 \  \  \mbox{$-1.50$} \ \  0.91 \ \  \mbox{$-0.07$} \ \  2.72 \ \  1.70 \ \  1.52 \  \ 1.66
    ].
\end{align*}
The performance of the controller $K_p$ is compared with $K_r$ which is designed using the S-variable approach from \cite[Theorem (iii-c)]{Ebihara2015s}. As shown in Fig. \ref{fig:costCom_pra}, for different randomly generated initial states, the cost of $K_p$ is consistently smaller than that of $K_r$ across all $\xi \in [-1,1]$, demonstrating the efficacy of the proposed PO algorithm.

The computational time and cost of PO is compared with the conventional BMI solvers PENLAB, PF, and CCD in Tables \ref{tab:TimeMassSpring} and \ref{tab:ComCostMassSpring}. The PO algorithm has substantially higher computational efficiency. As the order $N$ increases, PO has the lowest mean computational time across all orders, especially for higher orders as $N=8$, where PENLAB exceeds the time threshold, and PF and CCD take substantially longer.

\begin{figure}[thb]
    \centering
    \begin{subfigure}[b]{0.5\linewidth}
        \centering
        \includegraphics[width=\linewidth]{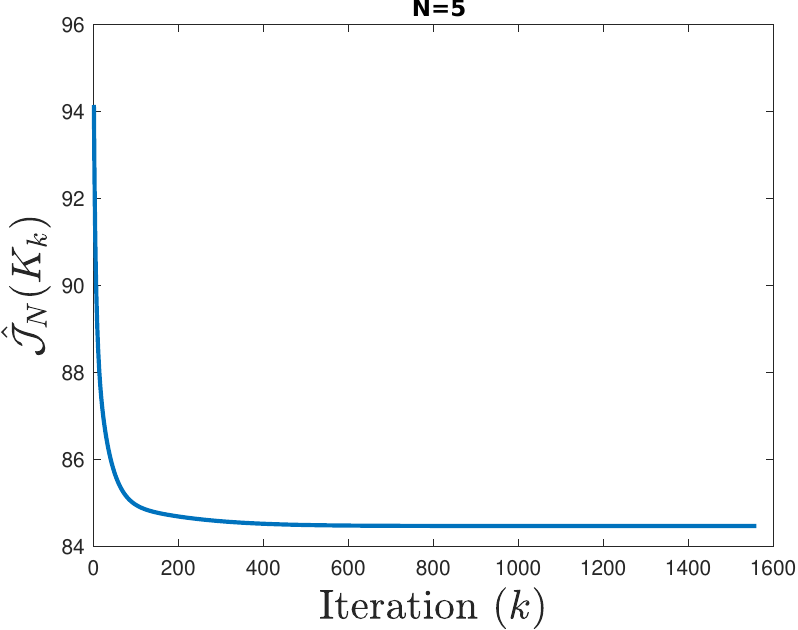}
    \end{subfigure}%
    \begin{subfigure}[b]{0.5\linewidth}
        \centering
        \includegraphics[width=\linewidth]{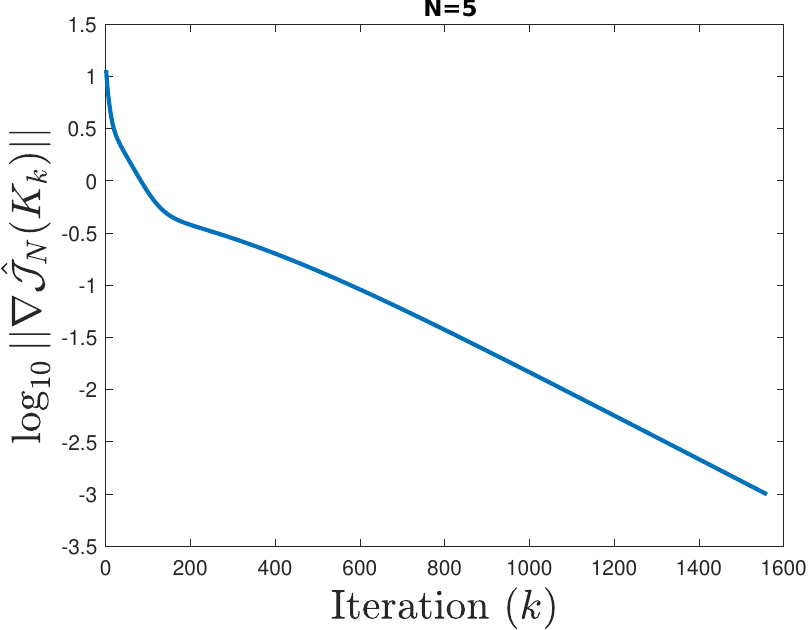}    
    \end{subfigure}%
    
\vspace{-0.2cm}
    
    \caption{For the mass-spring system, the evolution of the surrogate cost $\hat{\pazocal{J}}_N(K_k)$ and gradient $\nabla \hat{\pazocal{J}}_N(K_k)$ for the fifth-order of orthonormal polynomials.}
    \label{fig:Jconv_pra}
\end{figure}

\begin{figure}[thb]
    \centering
    \begin{subfigure}[b]{0.49\linewidth}
        \centering
        \includegraphics[width=\linewidth]{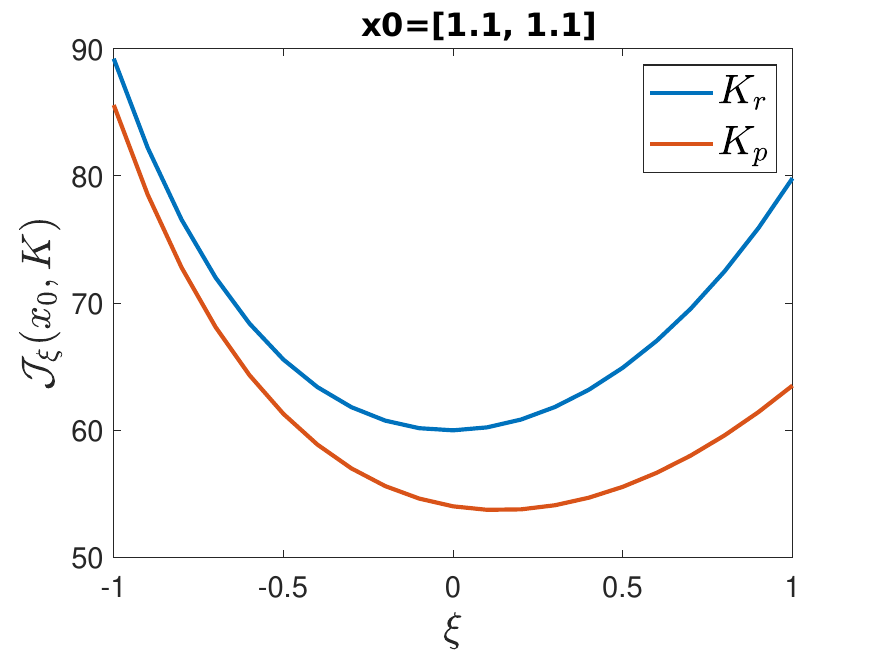}
    \end{subfigure}%
    \begin{subfigure}[b]{0.49\linewidth}
        \centering
        \includegraphics[width=\linewidth]{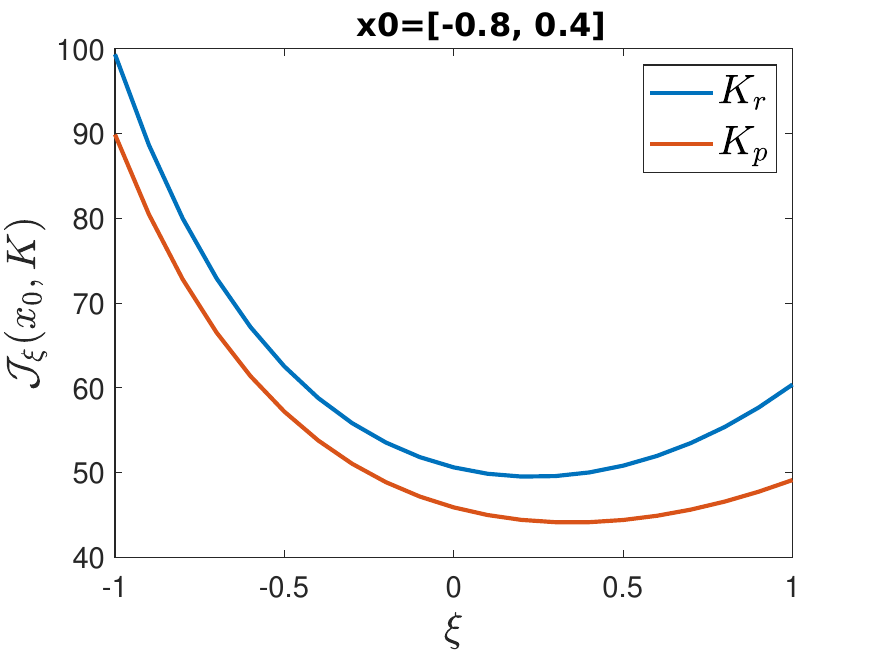}    
    \end{subfigure}%
    \\ 
    \begin{subfigure}[b]{0.49\linewidth}
        \centering
        \includegraphics[width=\linewidth]{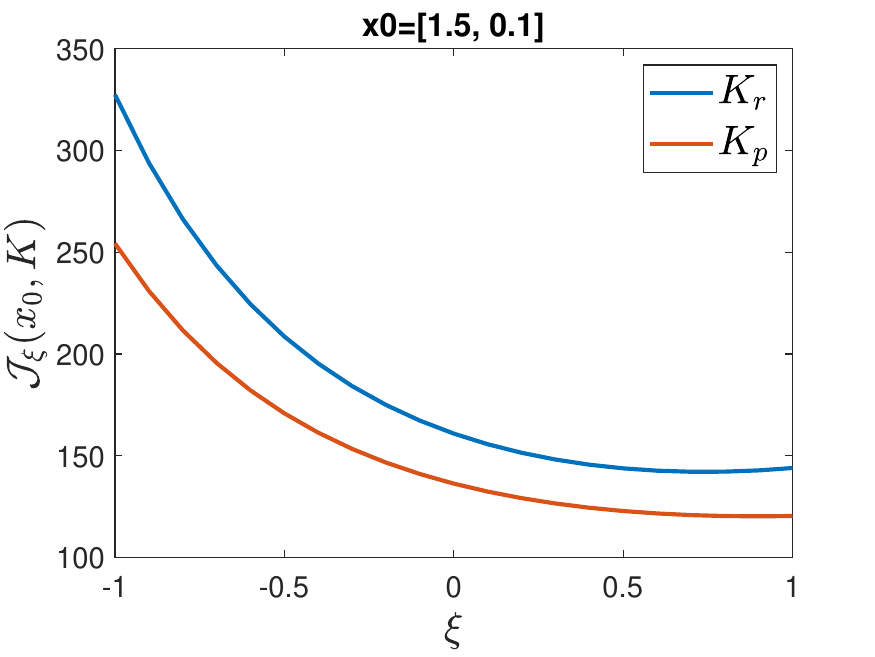}%
    \end{subfigure}%
    \begin{subfigure}[b]{0.49\linewidth}
    \centering
    \includegraphics[width=\linewidth]{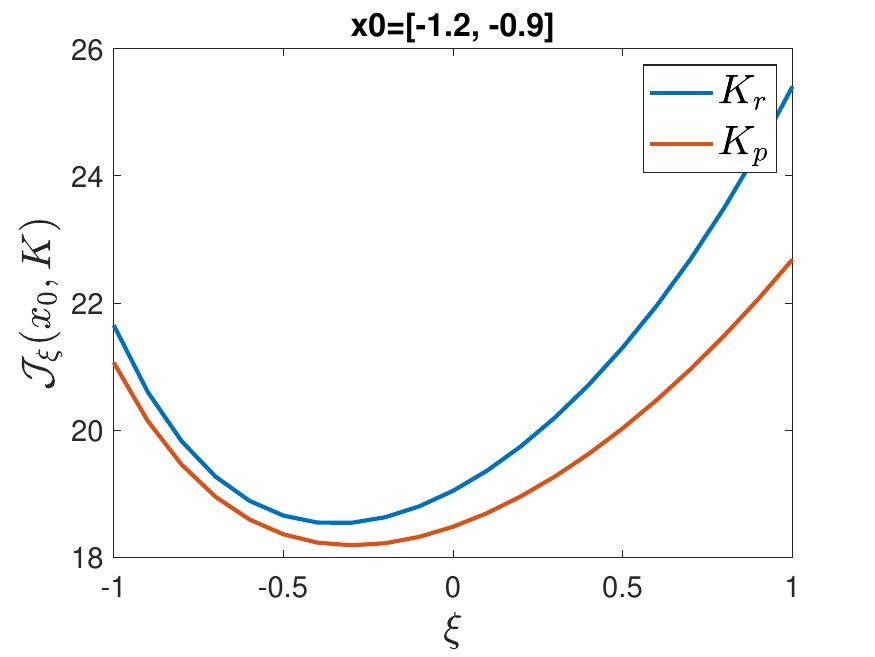}%
    \end{subfigure}%
    
\vspace{-0.2cm}
    
    \caption{For the mass-spring system, the cost for the controller $K_r$ generated by the S-variable approach \cite[Theorem (iii-c)]{Ebihara2015s} and the optimized controller $K_p$ by the PO algorithm in \eqref{eq:gradientDescent}.}
    \label{fig:costCom_pra}
\end{figure}

\begin{table*}[ht]
\centering
\caption{Comparison of computational times for the mass-spring system (unit: second, Alg.\ = Algorithm) }
\adjustbox{max width=\textwidth}{%
    \begin{tabular}{|c||c c|c c|c c|c c|}
        \hline
        \multirow{2}{6.43em}{\diagbox{Order}{Alg.}} & \multicolumn{2}{c|}{PO} & \multicolumn{2}{c|}{PENLAB} & \multicolumn{2}{c|}{PF} & \multicolumn{2}{c|}{CCD} \\ \cline{2-9}
                               & Mean & Std & Mean & Std & Mean & Std & Mean & Std \\ \hline \hline
        N=3 & 2.98  & 0.15  & 665.22  & 3.77  & 11.39  & 0.66  &22.85  & 0.47 \\ \hline
        N=5 & 6.33 & 0.25  & 5148.25  & 15.61  & 31.13  & 0.41  & 82.19  & 1.69 \\ \hline
        N=8 & 13.25  & 0.32  & $>10^{5}$  & 1754.70  & 101.57  & 2.29  &   520.09  & 3.46 \\ \hline
    \end{tabular}
}
\label{tab:TimeMassSpring}
\end{table*}

\begin{table}[thb]
    \caption{Comparison of cost $\hat{\pazocal{J}}_N(K)$ for the mass-spring system (Alg.\ = Algorithm)}
    \centering
    \begin{tabular}{|c||c|c|c|c|}
        \hline
        \diagbox{Order}{Alg.} &{PO} &{PENLAB} &{PF} &{CCD} \\ \hline \hline
        $N=3$ & 84.46  & 84.54   & 86.44  & 84.74 \\ \hline
        $N=5$ & 84.47  & 84.81  & 86.44 & 84.79  \\ \hline
        $N=8$ & 84.47   & *    & 86.44  & 84.82  \\ \hline
    \end{tabular}
    \label{tab:ComCostMassSpring}
\end{table}
\section{Conclusions}
This paper proposes a computationally efficient PO algorithm for solving the LQR problem for systems with probabilistic parameters. The algorithm was shown to converge to a stationary point at a linear convergence rate. In two case studies, the proposed PO algorithm is more computationally efficient than conventional BMI-based approaches and can identify a controller whose performance exceeds that of traditional robust control methods. Future work will focus on accelerating the algorithm by leveraging the geometric structure and second-order information of the LQR cost, aiming to enhance its efficiency and performance further.

\appendices
\section{Auxiliary Results on Matrices Manipulation}
This appendix summarizes fundamental properties of Lyapunov equations, trace-related inequalities, and Cauchy–Schwarz inequalities for random matrices.
\begin{lem}(Theorem 18 in \cite{book_sontag})\label{lm:LyaEquaIntegral}
    If $A \in \mathbb{R}^{n \times n}$ is Hurwitz, then the Lyapunov equation 
    \begin{align*}
        A^\top P + PA + Q = 0
    \end{align*}
    has a unique solution for any $Q \in \mathbb{R}^{n \times n}$, and the solution can be written as
    \begin{align*}
        P = \int_{0}^\infty e^{A^\top t} Q e^{At} \mathrm{d}t.
    \end{align*}
\end{lem}
\vspace{0.1cm}

\begin{cor}\label{cor:LyapPairtrace}
Suppose that $A \in \R^{n \times n}$ is Hurwitz, $M,N \in \Sym^n$, and $P,Y \in \Sym^n$ are the solutions of
\begin{subequations}
\begin{align}
    A^\top P + PA + M &= 0, \nonumber\\
    AY + YA^\top + N &= 0. \nonumber
\end{align}    
\end{subequations}
Then $\Tr[MY] = \Tr[NP]$.
\end{cor}

\begin{cor}\label{cor:LyapCompare}
Suppose that $A \in \R^{n \times n}$ is Hurwitz, $M_1, M_2 \in \Sym^{n}$ with $Q_2 \succeq Q_1$, and $P_1,P_2 \in \Sym^n$ are the solutions of
\begin{subequations}
\begin{align*}
    A^\top P_1 + P_1A + Q_1 &= 0 \\
    A^\top P_2 + P_2A + Q_2 &= 0.
\end{align*}    
\end{subequations}
Then $P_2 \succeq P_1$.
\end{cor}

\begin{lem}\label{lem:min-eig-PY}
Let $A\in\mathbb{R}^{n\times n}$ be Hurwitz and let $Q\in\mathbb{S}_{++}^{n}$.
Suppose there exist matrices $P,Y\in\mathbb{S}_{++}^{n}$ satisfying
\begin{subequations}
\begin{align}
A^\top P + P A + Q \preceq 0, \label{eq:lyapP}\\
A Y + Y A^\top + Q \preceq 0. \label{eq:lyapY}
\end{align}
\end{subequations}
Then
\[
\lambda_{\min}(P)\ \ge\ \frac{\lambda_{\min}(Q)}{2\|A\|},
\qquad
\lambda_{\min}(Y)\ \ge\ \frac{\lambda_{\min}(Q)}{2\|A\|}.
\]
\end{lem}
\vspace{0.1cm}
\begin{proof}
Let $v\in\mathbb{R}^{n}$ be a unit eigenvector of $P$ associated with
$\lambda_{\min}(P)$. Pre- and post-multiplying \eqref{eq:lyapP} by $v^\top$ and $v$
gives
\[
\lambda_{\min}(P)\, v^\top(A^\top + A)v \le - v^\top Q v \le -\lambda_{\min}(Q).
\]
Taking absolute values and using the bound
\[
|v^\top(A^\top + A)v| \le \|A^\top + A\| \le 2\|A\|
\]
gives that
\[
 2\|A\|\lambda_{\min}(P) \ge \lambda_{\min}(Q).
\]

The bound for $\lambda_{\min}(Y)$ follows by the same argument applied to
\eqref{eq:lyapY}. 
\end{proof}
    
\begin{lem}\label{lm:ExpDecay}
Let $A\in\mathbb{R}^{n\times n}$ be Hurwitz. Suppose there exist matrices
$P\in\mathbb{S}_{++}^{n_x}$ and $Q\in\mathbb{S}_{++}^{n_x}$ such that
\[
A^\top P + P A + Q \preceq 0.
\]
Consider the system $\dot x(t)=Ax(t)$ with $x(0)=x_0$. Then, for all $t\ge 0$,
\[
\|x(t)\|\le \mu\, e^{-\gamma t}\,\|x_0\|,
\]
where $\mu := \sqrt{\tfrac{\lambda_{\max}(P)}{\lambda_{\min}(P)}}$, and $\gamma := \tfrac{\lambda_{\min}(Q)}{2\,\lambda_{\max}(P)}$.
\end{lem}
\begin{proof}
Define \(V(x)=x^\top P x\); since
\(\lambda_{\min}(P)\|x\|^2\le V(x)\le \lambda_{\max}(P)\|x\|^2\) and
\(\dot V(x)\le -\lambda_{\min}(Q)\|x\|^2\le -\frac{\lambda_{\min}(Q)}{\lambda_{\max}(P)}V(x)\),
Grönwall’s inequality implies the result.
\end{proof}

\begin{lem}\label{lm:completeSquare}
For any $X,Y\in\mathbb{R}^{m\times n}$ and $\alpha>0$,
\[
X^\top Y+Y^\top X \preceq \alpha X^\top X+\frac{1}{\alpha}Y^\top Y.
\]
\end{lem}
\begin{proof}
$0\preceq \alpha\!\left(X-\alpha^{-1}Y\right)^\top\!\left(X-\alpha^{-1}Y\right)$.
\end{proof}

\begin{lem}(Trace Inequality \cite{Wang1986})\label{lm:traceIneq}
    Let $S \in \Sym^{n}$ and $P \in \Sym_+^{n}$. Then
    \begin{align}\label{eq:}
        \eigmin{S} \Tr(P) \le \Tr(SP) \le \eigmax{S} \Tr(P).
    \end{align}
\end{lem}

\begin{lem}(Cauchy-Schwarz Inequality for Random Matrices)\label{lm:CSIneqMatrix}
Let $X\in\mathbb{R}^{n\times m}$ and $Y\in\mathbb{R}^{q\times m}$ be square-integrable random matrices, and assume $\E[YY^\top]\succ 0$. Then
\vspace{-0.2cm}
\[
\|\E[XY^\top]\|\ \le\ \|\E[XX^\top]\|^{1/2}\,\|\E[YY^\top]\|^{1/2}.
\]
\end{lem}
\vspace{0.1cm}
\begin{proof}
For any $A\in \R^{n \times m}$ and $B \in \R^{q \times m}$, let $J(A,B) = \Tr\!\left[\E[(X^\top A + Y^\top B)(A^\top X + B^\top Y)] \right]$ and $B^* = \argmin_{B} J(A,B)$. Then
\begin{align}
    B^* = -\E[YY^\top]^{-1}\E[YX^\top]A,
\end{align}
since it is the zero of 
\begin{align}
    \nabla_B J(A,B) = 2\E[YY^\top]B + 2\E[YX^\top]A.
\end{align}
Plugging $B^*$ into $J(A,B)$ yields
\begin{align*}
&J(A,B^*)= 
\\ & \Tr\Big(
A^\top
\big(
\E[XX^\top]
-\E[XY^\top]
 \E[YY^\top]^{-1}\;\E[YX^\top]
\big)
A
\Big) \ge 0. 
\end{align*}
Since the inequality holds for any $A \in \R^{n \times m}$, it follows that
\begin{align}
    \E[XX^\top] - \E[XY^\top]\E[YY^\top]^{-1}\E[YX^\top] \succeq 0,
\end{align}
which consequently implies that
\begin{align}\label{eq:CSIDerivation}
\begin{split}
    &\norm{\E[XX^\top]} \ge \norm{\E[XY^\top]\E[YY^\top]^{-1}\E[YX^\top]} \\
    &\qquad\qquad\ \ \ge \norm{\E[XY^\top]\E[YX^\top]}/\norm{\E[YY^\top]}.
\end{split}
\end{align}
Therefore, the lemma readily follows from \eqref{eq:CSIDerivation}.
\end{proof}

\begin{lem}\label{lm:PYupperBound}
    Let $K$ be any stabilizing state-feedback gain for $(A,B)$, and let
    $P(K)\in\mathbb{S}_{++}^{n}$ and $Y(K)\in\mathbb{S}_{++}^{n}$ be the unique
    solutions to
\begin{subequations}\label{eq:PKYK}
    \begin{align}
    &(A-BK)^\top P(K) + P(K)(A-BK) + Q + K^\top R K = 0, \label{eq:PK}\\
    &(A-BK)Y(K) + Y(K)(A-BK)^\top + I_n = 0, \label{eq:YK}
    \end{align}
    \end{subequations}
    where $Q,R\in\mathbb{S}_{++}^{n}$. Then
    \begin{align}
    \begin{split}
        &\|P(K)\| \le \normF{P(K)} \le \Tr(P(K)), \\
        &\|Y(K)\| \le \normF{Y(K)} \le \Tr(Y(K)) \le \frac{\Tr(P(K))}{\lambda_{\min}(Q)}
    \end{split}    
    \end{align}

\end{lem}
    \vspace{0.1cm}
\begin{proof}
    The proof follows directly from the inequality
    \vspace{-0.2cm}
    \begin{align}
    \begin{split}
        \eigmin{Q}\Tr(Y(K)) &\le \Tr(Y(K)(Q+K^\top R K)) \\
        &= \Tr(P(K)),
    \end{split}
    \end{align}
        
        \vspace{-0.2cm}
        
\noindent
where the first inequality is due to Lemma~\ref{lm:traceIneq}, and the equality follows from Corollary~\ref{cor:LyapPairtrace}.
\end{proof}

\begin{lem}\label{lm:boundKbyJK}
Let $K$ be any stabilizing feedback gain for $(A,B)$, and let $P(K)\in\mathbb{S}_{++}^n$ be the solution of \eqref{eq:PK}.
Then
\[
\normF{K} \le a_1\Tr(P(K)) + a_2\sqrt{\Tr(P(K))},
\]
where $a_1:=\tfrac{2\|B\|}{\lambda_{\min}(R)}$ and
$a_2:=\sqrt{\tfrac{2\|A\|}{\lambda_{\min}(R)}}$.
\end{lem}
\begin{proof}
    By Lemma \ref{lem:min-eig-PY} with $A$ replaced by $A-BK$, the smallest eigenvalue of $Y(K)$ (solution of \eqref{eq:YK}) admits the lower bound:
    \begin{align}\label{eq:eigminY}
    \begin{split}
        &\eigmin{Y} \ge {1}/{(2\norm{A} + 2\norm{B}\normF{K})}.
    \end{split}
    \end{align}
    Applying Corollary \ref{cor:LyapPairtrace} and using the trace inequality in Lemma \ref{lm:traceIneq} and \eqref{eq:eigminY} result in
    \begin{align}\label{eq:JxiLowBound}
    \begin{split}
       \Tr(P(K)) &= \Tr[Y(K)(Q+K^\top R K)] \\
        &\ge \frac{\eigmin{R}\normF{K}^2}{2\norm{A} + 2\norm{B}\normF{K}}.
    \end{split}
    \end{align}
    Arranging \eqref{eq:JxiLowBound} as a quadratic inequality with respective to $\normF{K}$ and bounding the largest root of the associated equality result in
    \begin{align*}
    \begin{split}
        &\normF{K} \le \frac{\norm{B}\Tr(P(K))}{\eigmin{R}}   \\
        &\quad + \frac{ \left[\norm{B}^2{\Tr(P(K))}^2 + 2\eigmin{R}\norm{A}\Tr(P(K))\right]^{{1}/{2}}}{\eigmin{R}}
    \end{split}
    \end{align*}
    The proof is completed by the subadditivity of the square-root function.
\end{proof}

\begin{lem}\label{lm:normApceBound}
Let $A(\xi)\in\mathbb{R}^{n\times m}$ be continuous on a compact set $\Xi$, and let
$\phi_N(\xi)\in\mathbb{R}^{N}$ satisfy $\E[\phi_N(\xi)\phi_N(\xi)^\top]=I_N$. Then, for all $N\in\mathbb{Z}_+$,
\[
\Big\|\E\!\big[(\phi_N(\xi)\otimes I_n)\,A(\xi)\,(\phi_N(\xi)^\top\otimes I_m)\big]\Big\|
\le \sup_{\xi\in\Xi}\|A(\xi)\|.
\]
\end{lem}
\begin{proof}
By the Cauchy-Schwarz inequality in Lemma \ref{lm:CSIneqMatrix},
\begin{align*}
\begin{split}
    &\big\|{\E\big[(\phi_N(\xi) \otimes I_n) A(\xi) (\phi_N(\xi)^\top \otimes I_m) \big]}\big\| \\
    &\qquad\le \big\|\E\big[ (\phi_N(\xi) \otimes I_n)A(\xi)A(\xi)^\top(\phi_N(\xi)^\top \otimes I_n) \big] \big\|^{{1}/{2}} \\
    &\qquad\quad\ \times \big\|\E\big[ (\phi_N(\xi) \otimes I_m)(\phi_N(\xi)^\top \otimes I_m) \big]\big\|^{{1}/{2}} \\
    &\qquad\le \big\|\bar{a}^2\E\left[ \phi_N(\xi)\phi_N(\xi)^\top \otimes I_n\right]\big\|^{{1}/{2}} = \bar{a}
\end{split}
\end{align*}
where $\bar{a} = \sup_{\xi \in \Xi}  \norm{A(\xi)}$. 
\end{proof}

\section{Supporting Lemmas for Theorem \ref{thm:gradientConv}}
This appendix introduces several lemmas to support the proof of the main theorem. The following theoretical results are developed based on Assumption \ref{ass:heterogeneity}.
\begin{lem}\label{lm:HeteboundStabilizing}
Under Assumption~\ref{ass:heterogeneity}, it holds that
    \vspace{-0.2cm}
\[
\pazocal{S}(\bar{\xi},h) \subseteq \hat{\pazocal{S}}_N,
\qquad
\pazocal{S}(\bar{\xi},h) \subseteq \pazocal{S}(\xi),\ \ \forall\,\xi\in\Xi.
\]
\end{lem}
    \vspace{0.1cm}
\begin{proof}
\emph{($\pazocal{S}(\bar{\xi},h) \subseteq \hat{\pazocal{S}}_N$):}
Let $K\in\pazocal{S}(\bar{\xi},h)$. Using \eqref{eq:AbarBbarLyapuP} and the definitions of
$\tilde{\pazocal{A}}_N$ and $\tilde{\pazocal{B}}_N$ in \eqref{eq:parAtildeBtildeDef} gives that
\begin{align}\label{eq:ABPbar}
\begin{split}
    &\pazocal{A}_{N,c}(K)^\top \bar{\pazocal{P}}_N(K) + \bar{\pazocal{P}}_N(K)\pazocal{A}_{N,c}(K)    \\
    &\quad + I_{N+1}\otimes(Q+K^\top R K) \\
    &\quad - \big(\tilde{\pazocal{A}}_N-\tilde{\pazocal{B}}_N(I_{N+1}\otimes K)\big)^\top
    \bar{\pazocal{P}}_N(K) \\
    &\quad - \bar{\pazocal{P}}_N(K)
    \big(\tilde{\pazocal{A}}_N-\tilde{\pazocal{B}}_N(I_{N+1}\otimes K)\big)=0.
\end{split}
\end{align}
By Lemmas \ref{lm:PYupperBound},~\ref{lm:boundKbyJK}, and~\ref{lm:normApceBound}, the perturbation terms satisfy
\begin{align*}
&\big\| \big(\tilde{\pazocal{A}}_N-\tilde{\pazocal{B}}_N(I_{N+1}\otimes K)\big)^\top \bar{\pazocal{P}}_N(K) \big\| \\
&\le h(1+\|K\|)\epsilon(h) \\
&\le h\big(1+a_1(\bar{\xi})h+a_2(\bar{\xi})\sqrt{h}\big)\epsilon(h) \le {\lambda_{\min}(Q)}/{4}.    
\end{align*}

Hence, the sum of the last three terms in \eqref{eq:ABPbar} is positive definite. By
\cite[Lemma~3.19]{book_zhou}, this implies that
$\pazocal{A}_N-\pazocal{B}_N(I_{N+1}\otimes K)$ is Hurwitz, and therefore
$K\in\hat{\pazocal{S}}_N$.

\medskip
\emph{($\pazocal{S}(\bar{\xi},h) \subseteq \pazocal{S}(\xi)$):}
The Lyapunov equation \eqref{eq:LyapuP} evaluated at $\bar{\xi}$ can be rewritten as
\begin{align}\label{eq:ABKPbarxi}
\begin{split}
    &A_c(K,\xi)^\top P(K,\bar{\xi}) + P(K,\bar{\xi})A_c(K,\xi) + Q+K^\top R K\\
    &\, + \big[(A(\bar{\xi})-A(\xi))-(B(\bar{\xi})-B(\xi))K\big]^\top P(K,\bar{\xi}) \\
    &\, + P(K,\bar{\xi})\big[(A(\bar{\xi})-A(\xi))-(B(\bar{\xi})-B(\xi))K\big] = 0.
\end{split}
\end{align}
Using Lemmas \ref{lm:PYupperBound} and~\ref{lm:boundKbyJK}, the perturbation term is bounded as
\begin{align}\label{eq:Perrbound}
\begin{split}
    &\big\|P(\bar{\xi},K)\big[(A(\bar{\xi})-A(\xi))-(B(\bar{\xi})-B(\xi))K\big]\big\|  \\
    &\quad \le {\lambda_{\min}(Q)}/{4}.
\end{split}
\end{align}
Thus, the sum of the last three terms in \eqref{eq:ABKPbarxi} is positive definite. Applying \cite[Lemma~3.19]{book_zhou} again shows that $A(\xi)-B(\xi)K$ is Hurwitz for all $\xi\in\Xi$, which completes the proof.
\end{proof}

The following lemma bounds the deviation between $(\pazocal{P}_N,\pazocal{Y}_N)$ and
$(\bar{\pazocal{P}}_N,\bar{\pazocal{Y}}_N)$.

\begin{lem}\label{lm:PYLyapuheter}
Under Assumption~\ref{ass:heterogeneity}, it holds that
\begin{align}\label{eq:PYLyapuheter}
\begin{split}
\|\pazocal{P}_N(K)-\bar{\pazocal{P}}_N(K)\| &\le \epsilon_1,\, \|\pazocal{Y}_N(K)-\bar{\pazocal{Y}}_N(K)\| \le \epsilon_1,
\end{split}
\end{align}
for all $K\in\pazocal{S}(\bar{\xi},h)$.
\end{lem}

\begin{proof}
By Lemma~\ref{lm:HeteboundStabilizing}, any $K\in\pazocal{S}(\bar{\xi},h)$ satisfies
$K\in\hat{\pazocal{S}}_N$, so both Lyapunov equations are well posed.

Subtracting \eqref{eq:LyapunovPPCE} from \eqref{eq:AbarBbarLyapuP} gives
\begin{align}
\begin{split}
&\bar{\pazocal{A}}_{N,c}(K)^\top\Delta\pazocal{P}_N(K)
+\Delta\pazocal{P}_N(K)\,\bar{\pazocal{A}}_{N,c}(K) \\
&\quad -\tilde{\pazocal{A}}_{N,c}(K)^\top \pazocal{P}_N(K)
-\pazocal{P}_N(K)\,\tilde{\pazocal{A}}_{N,c}(K)=0,
\end{split}
\end{align}
where $\Delta\pazocal{P}_N(K):=\bar{\pazocal{P}}_N(K)-\pazocal{P}_N(K)$, and
$\tilde{\pazocal{A}}_{N,c}(K):=\tilde{\pazocal{A}}_N-\tilde{\pazocal{B}}_N(I_{N+1}\otimes K)$.
By Lemma~\ref{lm:LyaEquaIntegral},
\begin{align*}
\begin{split}
    &\Delta\pazocal{P}_N(K)
    = -\int_0^\infty e^{\bar{\pazocal{A}}_{N,c}(K)^\top t}\Big(\tilde{\pazocal{A}}_{N,c}(K)^\top \pazocal{P}_N(K) \\
    &\quad +\pazocal{P}_N(K)\tilde{\pazocal{A}}_{N,c}(K)\Big)
    e^{\bar{\pazocal{A}}_{N,c}(K)t}\,\de t.    
\end{split}    
\end{align*}
Hence,
\begin{align}\label{eq:DeltaP-bound}
&\|\Delta\pazocal{P}_N(K)\| \le 2\|\tilde{\pazocal{A}}_{N,c}(K)\|\,\|\pazocal{P}_N(K)\|\,\nonumber
\\
&\qquad \times \Big\|\int_0^\infty e^{\bar{\pazocal{A}}_{N,c}(K)^\top t}
e^{\bar{\pazocal{A}}_{N,c}(K)t}\,\de t\Big\| \nonumber\\
&\quad\le \frac{2}{\lambda_{\min}(Q)}\,\|\tilde{\pazocal{A}}_{N,c}(K)\|\,\|\pazocal{P}_N(K)\|\,
\|\bar{\pazocal{P}}_N(K)\|,
\end{align}
where the last step follows from Corollary~\ref{cor:LyapCompare}, i.e., $\int_0^\infty e^{\bar{\pazocal{A}}_{N,c}^\top t}e^{\bar{\pazocal{A}}_{N,c}t}dt \preceq \lambda_{\min}(Q)^{-1}\bar{\pazocal{P}}_N(K)$. Moreover, by Assumption~\ref{ass:heterogeneity} and Lemma~\ref{lm:normApceBound},
\[
\|\tilde{\pazocal{A}}_{N,c}(K)\|
\le \|\tilde{\pazocal{A}}_N\|+\|\tilde{\pazocal{B}}_N\|\,\|K\|
\le (1+\|K\|)\epsilon(h),
\]
and Lemma~\ref{lm:boundKbyJK} yields $\|K\|\le \|K\|_{\mathrm F}\le a_1h+a_2\sqrt{h}$. Using $\|\bar{\pazocal{P}}_N(K)\|\le h$ (Lemma~\ref{lm:PYupperBound}) gives that
\begin{align*}
    &\|\Delta\pazocal{P}_N(K)\| \le \frac{2h \norm{\pazocal{P}_N(K)}}{\lambda_{\min}(Q)}\big(1+a_1 h+a_2 \sqrt{h}\big)\epsilon(h) \\
    &\quad \le \big(h + \|\Delta\pazocal{P}_N(K)\|\big) \min\{1/2, \epsilon_1/(2h) \}
\end{align*}
where the last inequality follows from the definition of $\epsilon(h)$ in Assumption~\ref{ass:heterogeneity}. This proves the first bound in \eqref{eq:PYLyapuheter}.

\smallskip
\noindent\emph{Bound on $\pazocal{Y}_N(K)-\bar{\pazocal{Y}}_N(K)$.}
An analogous argument applies to \eqref{eq:LyapunovYPCE} and \eqref{eq:AbarBbarLyapuY}. Let $\Delta\pazocal{Y}_N(K):=\bar{\pazocal{Y}}_N(K)-\pazocal{Y}_N(K)$. Subtracting the two Lyapunov equations yields
\begin{align*}
    &\bar{\pazocal{A}}_{N,c}(K)\Delta\pazocal{Y}_N(K)
    +\Delta\pazocal{Y}_N(K)\bar{\pazocal{A}}_{N,c}(K)^\top \\
    &+\tilde{\pazocal{A}}_{N,c}(K)\pazocal{Y}_N(K)
    +\pazocal{Y}_N(K)\Delta\tilde{\pazocal{A}}_{N,c}(K)^\top=0.    
\end{align*}

Applying Lemma~\ref{lm:LyaEquaIntegral}, Corollary~\ref{cor:LyapCompare}, and the same bounds on $\|\tilde{\pazocal{A}}_{N,c}(K)\|$ gives
\begin{align*}
&\|\Delta\pazocal{Y}_N(K)\|
\le 2\|\tilde{\pazocal{A}}_{N,c}(K)\|\,\|\pazocal{Y}_N(K)\|\, \\
& \quad \times \Big\|\int_0^\infty e^{\bar{\pazocal{A}}_{N,c}(K)t}e^{\bar{\pazocal{A}}_{N,c}(K)^\top t}\,\de t\Big\| \\
& \le \big(h/\lambda_{\min}(Q) + \|\Delta\pazocal{Y}_N(K)\|\big)  \min\{1/2, \eigmin{Q}\epsilon_1/(2h)\}   
\end{align*}
where we additionally use $\|\bar{\pazocal{Y}}_N(K)\| \le h/\lambda_{\min}(Q)$ from Lemma~\ref{lm:PYupperBound} and the choice of $\epsilon(h)$ in Assumption~\ref{ass:heterogeneity}. This proves the second bound in \eqref{eq:PYLyapuheter}.
\end{proof}

The following lemma bounds $S(K)$ defined in \eqref{eq:SKdef}.

\begin{lem}\label{lm:boundSK}
Under Assumption \ref{ass:heterogeneity}, for all $K\in \pazocal{S}(\bar{\xi},h)$, it holds that $[\pazocal{A}_N - \pazocal{B}_N(I_{N+1} \otimes K)]_{0,0}$ is Hurwitz, and
\begin{align*}
    \|S(K)\|&\le h+\epsilon_1 \\
    \lambda_{\min}\!\big(S(K)\big)&\ge \Big[ 2\bar{a}+2\bar{b}\big(a_1h+a_2\sqrt{h}\big)\Big]^{-1} ,  
\end{align*}
where $\bar{a} = \sum_{\xi \in \Xi}\norm{A(\xi)}$ and $\bar{b} = \sum_{\xi \in \Xi}\norm{B(\xi)}$.
\end{lem}

\begin{proof}
Fix any $K\in \pazocal{S}(h,\bar{\xi})$ and define
\begin{align*}
\Delta_1(K,\bar{\xi}):&=A_c(K,\bar{\xi})-[\pazocal{A}_{N,c}(K)]_{0,0}.    
\end{align*}
By \eqref{eq:ABerror} and Lemma \ref{lm:boundKbyJK},
\begin{align}\label{eq:Delta1-bound}
\begin{split}
\|\Delta_1\| &\le (1+\|K\|)\epsilon(h) \le \big(1+a_1h + a_2\sqrt{h})\big) \epsilon(h).
\end{split}
\end{align}

\noindent\emph{Step 1: Bound $\|S(K)\|$.}
Subtracting \eqref{eq:LyapuY} (with parameter $\bar{\xi}$) from \eqref{eq:SKdef} yields
\begin{align}\label{eq:SK-diff}
\begin{split}
&A_c\big(S(K)-Y(K,\bar{\xi})\big)+\big(S(K)-Y(K,\bar{\xi})\big)A_c^\top  \\
&-\Delta_1(K,\bar{\xi})S(K)-S(K)\Delta_1(K,\bar{\xi})^\top=0.
\end{split}
\end{align}
Since $A_c(K,\bar{\xi})$ is Hurwitz, Lemma~\ref{lm:LyaEquaIntegral} implies
\begin{align*}
    S(K)-Y(K,\bar{\xi})
&=-\int_{0}^{\infty} e^{A_ct}\big(\Delta_1S +S\Delta_1^\top\big)e^{A_c^\top t}\,\de t
\end{align*}
and therefore
\begin{align}\label{eq:SKY-bound1}
\|S(K)-Y(K,\bar{\xi})\|
&\le 2\|\Delta_1\|\,\|S\| \Big\|\int_{0}^{\infty} e^{A_ct}e^{A_c^\top t}\,dt\Big\|.
\end{align}
By Lemma~\ref{lm:LyaEquaIntegral} and the Lyapunov equation
$A_cY(K,\bar{\xi})+Y(K,\bar{\xi})A_c^\top+I=0$, we have
\begin{equation}\label{eq:int-bound-by-Y}
\Big\|\int_{0}^{\infty} e^{A_c(K,\bar{\xi})t}e^{A_c(K,\bar{\xi})^\top t}\,dt\Big\| = \|Y(K,\bar{\xi})\|.
\end{equation}
Moreover, by Lemma~\ref{lm:PYupperBound} and $K\in\pazocal{S}(\bar{\xi},h)$,
\begin{equation}\label{eq:Ybar-bound}
\|Y(K,\bar{\xi})\|\le  {h}/{\lambda_{\min}(Q)}.
\end{equation}
Combining \eqref{eq:SKY-bound1}--\eqref{eq:Ybar-bound} and using Assumption~\ref{ass:heterogeneity} give
\begin{align*}
    \|S-Y\|
    &\le \frac{2h}{\lambda_{\min}(Q)}\,\|\Delta_1\|\,  \big(\|Y\|+\|S-Y\|\big) \\
    &\le \big(h + \|S-Y\|\big) \min\{1/2, \epsilon_1/(2h) \}. 
\end{align*}
where we used $\|Y(K,\bar{\xi})\|\le h$. The first statement is thus proved.

\noindent\emph{Step 2: Bound $\lambda_{\min}(S(K))$.}
From the identity $A_c=[\pazocal{A}_{N,c}(K)]_{0,0}+\Delta_1$, we can rewrite
\begin{align}\label{eq:Y-perturbed}
\begin{split}
&[\pazocal{A}_{N,c}(K)]_{0,0}Y(K,\bar{\xi})+Y(K,\bar{\xi})[\pazocal{A}_{N,c}(K)]_{0,0}^\top \\
&\quad +\Delta_1Y(K,\bar{\xi})+Y(K,\bar{\xi})\Delta_1^\top + I_{n_x}=0.
\end{split}
\end{align}
By \eqref{eq:Delta1-bound}, Lemma \ref{lm:PYupperBound} and  Assumption~\ref{ass:heterogeneity},
\begin{align*}
\begin{split}
\|\Delta_1\|\,\|Y(K,\bar{\xi})\|
&\le \frac{h}{\lambda_{\min}(Q)}\big(1+a_1h+a_2\sqrt{h}\big)\epsilon(h)
\\
&\le {1}/{4}.    
\end{split}
\end{align*}
Thus the last three terms in \eqref{eq:Y-perturbed} sum to a positive definite matrix, and \cite[Lemma~3.19]{book_zhou} implies that $[\pazocal{A}_{N,c}(K)]_{0,0}$ is Hurwitz. Applying Lemma~\ref{lem:min-eig-PY} to the Lyapunov equation defining $S(K)$ yields
\begin{align*}
\begin{split}
    \lambda_{\min}\!\big(S(K)\big) &\ge \frac{1}{2\|[\pazocal{A}_{N,c}(K)]_{0,0}\|} \\
    &\ge \frac{1}{2\|[\pazocal{A}_{N}]_{0,0}\|+2\|[\pazocal{B}_{N}]_{0,0}\|\,\|K\|}.    
\end{split}    
\end{align*}

Finally, $\|[\pazocal{A}_{N}]_{0,0}\|\le \|\pazocal{A}_N\| \le \bar{a}$,
$\|[\pazocal{B}_{N}]_{0,0}\|\le \|\pazocal{B}_N\| \le \bar{b}$, and
$\|K\|\le \|K\|_{\mathrm F}\le a_1h+a_2\sqrt{h}$ (Lemma~\ref{lm:boundKbyJK}), which gives the stated bound.
\end{proof}

The following lemma shows that, under Assumption \ref{ass:heterogeneity}, for any fixed $\bar{\xi}\in\Xi$, the gradient mismatch between the surrogate-model cost $\hat{\pazocal{J}}_N$ in \eqref{eq:surrogateLTI} and the nominal cost $\pazocal{J}_{\bar{\xi}}$ associated with $(A(\bar{\xi}),B(\bar{\xi}))$ is uniformly small.

\begin{lem}\label{lm:gradientHet}
Under Assumption \ref{ass:heterogeneity}, for all $K \in \pazocal{S}(h,\bar{\xi})$, it holds that
\begin{align}\label{eq:boundGradientHet}
    \|\nabla \hat{\pazocal{J}}_N(K) - \nabla \pazocal{J}_{\bar{\xi}}(K)\|_{\mathrm F}
    \le n_x\, b_1(h)\,\epsilon_1,
\end{align}
where
\begin{align*}
b_1(h) &:= \frac{h+1}{2h\big(1+a_1h+a_2\sqrt{h}\big)}
+ \frac{2\bar{b}h}{\lambda_{\min}(Q)}  \\
&+{(h+1)\Big(\big(a_1\|R\|+\bar{b}\big)h+a_2\|R\|\sqrt{h}+\bar{b}\Big)}\big/{h} .
\end{align*}
\end{lem}
\begin{proof}
In the proof, we omit the arguments $K$ and $N$ whenever no confusion arises. It follows from \eqref{eq:IncrementLyapunovPPCE} that the Fr\'echet derivative
$\pazocal{P}' = \pazocal{P}'(K)[E]$ admits the integral representation
\begin{align}\label{eq:Pprime-int}
\begin{split}
\pazocal{P}'
&=\int_{0}^{\infty} e^{\pazocal{A}_{c}^\top t}
\Big( (I_{N+1}\otimes E^\top)\pazocal{E} \\
&\qquad +\pazocal{E}^\top (I_{N+1}\otimes E) \Big)
e^{\pazocal{A}_{c} t}\,\de t,
\end{split}
\end{align}
where $\pazocal{A}_{c}:=\pazocal{A}_{N,c}(K)$ and $\pazocal{E}:=\pazocal{E}_N(K)$. Taking norms in \eqref{eq:Pprime-int} yields
\begin{align}\label{eq:Pprime-bound}
\begin{split}
\|\pazocal{P}'\| \le 2\|E\|\,\|\pazocal{E}\|\,
\Big\|\int_{0}^{\infty} e^{\pazocal{A}_{c}^\top t}e^{\pazocal{A}_{c} t}\,\de t\Big\|.
\end{split}
\end{align}
Moreover, by Corollary~\ref{cor:LyapCompare} applied to the Lyapunov equation \eqref{eq:LyapunovPPCE},
\begin{equation}\label{eq:int-bound-P}
\int_{0}^{\infty} e^{\pazocal{A}_{c}^\top t}e^{\pazocal{A}_{c} t}\,\de t
\preceq \lambda_{\min}(Q)^{-1}\pazocal{P}_N(K).
\end{equation}
Substituting the norm of \eqref{eq:int-bound-P} into \eqref{eq:Pprime-bound} gives
\begin{equation}\label{eq:Pprime-bound2}
\|\pazocal{P}'\|
\le \frac{2\|E\|}{\lambda_{\min}(Q)}\,\|\pazocal{E}\|\,\|\pazocal{P}_N\|.
\end{equation}
Finally, using $\pazocal{E}=I_{N+1}\otimes (RK)-\pazocal{B}_N^\top \pazocal{P}_N$ and $\|I_{N+1}\otimes (RK)\| \le \|R\|\,\|K\|$ gives that
\begin{equation}\label{eq:Enorm-bound}
\|\pazocal{E}\|\le \|R\|\,\|K\|+\|\pazocal{B}_N\|\,\|\pazocal{P}_N\|.
\end{equation}
Combining \eqref{eq:Pprime-bound2}--\eqref{eq:Enorm-bound} and the bounds
$\|K\|\le a_1h+a_2\sqrt{h}$ (Lemma \ref{lm:boundKbyJK}), $\|\pazocal{P}_N\|\le h+\epsilon_1$ (Lemma \ref{eq:PYLyapuheter}) and $\norm{\pazocal{B}_N} \le \bar{b}$ yields
\begin{align}\label{eq:PderivBound}
\|\pazocal{P}'\|
\le
\frac{2\|E\|(h+\epsilon_1)}{\lambda_{\min}(Q)}
\Big((a_1\|R\|+\bar{b})h + a_2\norm{R}\sqrt{h}+\bar{b}\epsilon_1\Big).    
\end{align}

Next, subtract \eqref{eq:IncrementLyapunovPPCE} from its nominal counterpart (with
$\bar{\pazocal{P}}_N$, $\bar{\pazocal{A}}_N$, and $\bar{\pazocal{B}}_N$) to yield
\begin{align}\label{eq:DeltaPprime-Lyap}
\begin{split}
&\bar{\pazocal{A}}_{c}^\top(\bar{\pazocal{P}}'-\pazocal{P}')
+(\bar{\pazocal{P}}'-\pazocal{P}')\bar{\pazocal{A}}_{c} \\
&+\Big((I_{N+1}\otimes E^\top)(\bar{\pazocal{E}}-\pazocal{E})
+(\bar{\pazocal{E}}-\pazocal{E})^\top(I_{N+1}\otimes E)\Big) \\
&-\tilde{\pazocal{A}}_{c}^\top \pazocal{P}'-\pazocal{P}'\,\tilde{\pazocal{A}}_{c}=0,
\end{split}
\end{align}
where 
\[
\bar{\pazocal{A}}_{c}:=\bar{\pazocal{A}}_{N}-\bar{\pazocal{B}}_{N}(I_{N+1}\otimes K),\quad
\tilde{\pazocal{A}}_{c}:=\pazocal{A}_{c}-\bar{\pazocal{A}}_{c},
\]
\[
\bar{\pazocal{P}}':=\bar{\pazocal{P}}'_N(K)[E],\ \ \pazocal{P}':=\pazocal{P}'_N(K)[E].
\]
Since $\bar{\pazocal{A}}_{c}$ is Hurwitz, Lemma~\ref{lm:LyaEquaIntegral} and \eqref{eq:int-bound-P} applied to \eqref{eq:DeltaPprime-Lyap} gives
\begin{align}\label{eq:DeltaPprime-bound}
\begin{split}
\|\bar{\pazocal{P}}'-\pazocal{P}'\|
&\le \tfrac{2\|\bar{\pazocal{P}}_N\|}{\lambda_{\min}(Q)}
\Big(\|E\|\,\|\bar{\pazocal{E}}-\pazocal{E}\|+\|\tilde{\pazocal{A}}_{c}\|\,\|\pazocal{P}'\|\Big).
\end{split}
\end{align} 

For any $E\in\mathbb{R}^{n_u\times n_x}$,
\begin{align*}
&\big|\langle \nabla \hat{\pazocal{J}}_N(K)-\nabla \pazocal{J}_{\bar{\xi}}(K),E\rangle\big| =\big|\Tr\!\big((\bar{\pazocal{P}}'-\pazocal{P}')\pazocal{I}_N\pazocal{I}_N^\top\big)\big| \\
&=\big|\Tr\big([\bar{\pazocal{P}}'-\pazocal{P}']_{0,0}\big)\big|
\le n_x\,\|\bar{\pazocal{P}}'-\pazocal{P}'\|
\le n_x\,b_1(h)\,\epsilon_1\,\|E\|.
\end{align*}
Taking the supremum over $\|E\|_{\mathrm F}=1$ and using $\|E\|\le \|E\|_{\mathrm F}$ yields
\eqref{eq:boundGradientHet}.
\end{proof}

Under Assumption~\ref{ass:heterogeneity}, the following lemma shows that the gradient norm $\|\nabla \hat{\pazocal{J}}_N(K)\|_{\mathrm F}$ admits a lower bound in terms of $\|\pazocal{E}_N(K)\|_{\mathrm F}$, where $\pazocal{E}_N$ is defined in \eqref{eq:defpazocalE}.

\begin{lem}\label{lm:gradLowerBound}
Under Assumption~\ref{ass:heterogeneity}, for any $K\in\pazocal{S}(\bar{\xi},h)$,
\begin{align}\label{eq:gradLowerBound}
\|\nabla \hat{\pazocal{J}}_N(K)\|_{\mathrm F}
\ge b_2(h) \big\|[\pazocal{E}_N(K)]_{0,0}\big\|_{\mathrm F} - b_3(h)\,\epsilon_1,
\end{align}
where
\begin{align*}
    b_2(h) &=\Big[\bar{a}+ \bar{b}\big(a_1h+a_2\sqrt{h}\big) \Big]^{-1} \\
    b_3(h) &= {n_x(h+1)^2 }\big((a_1\|R\|+\bar{b})h + a_2\norm{R}\sqrt{h}+\bar{b}\epsilon_1\big)/{h^2} ,
\end{align*}
and $a_1,a_2$ are from Lemma~\ref{lm:boundKbyJK} evaluated at $(A(\bar{\xi}),B(\bar{\xi}))$.
\end{lem}

\begin{proof}
Omit the argument $K$ and $N$ when there is no ambiguity. Consider the $(0,0)$ block of \eqref{eq:IncrementLyapunovPPCE}. Replacing $\de \pazocal{P}_N$ by $\pazocal{P}' = \pazocal{P}_N'(K)[E]$ and $\de K$ by $E \in \R^{n_u \times n_x}$ yields
\begin{align}\label{eq:block00}
\begin{split}
&[\pazocal{A}_{c}]_{0,0}^\top [\pazocal{P}']_{0,0}+[\pazocal{P}']_{0,0}[\pazocal{A}_{c}]_{0,0} +E^\top[\pazocal{E}]_{0,0}+[\pazocal{E}]_{0,0}^\top E \\
&\quad +[\pazocal{A}_{c}]_{1\!:\!N,0}^\top[\pazocal{P}']_{1\!:\!N,0} +[\pazocal{P}']_{0,1\!:\!N}[\pazocal{A}_{c}]_{1\!:\!N,0}=0 .
\end{split}
\end{align}
Applying Corollary~\ref{cor:LyapPairtrace} to \eqref{eq:block00} and using \eqref{eq:SKdef} gives that
\begin{align}\label{eq:grad-decomp}
\begin{split}
\langle \nabla \hat{\pazocal{J}}_N(K),E\rangle
&=\Tr\big([\pazocal{P}'_N(K)[E]]_{0,0}\big)
= \pazocal{Q}_1(E)+\pazocal{Q}_2(E),
\end{split}
\end{align}
where, with $S:=S(K)$ (cf.\ \eqref{eq:SKdef}),
\begin{align}\label{eq:Q1def}
\pazocal{Q}_1(E)
:=2\Tr\big(E^\top[\pazocal{E}]_{0,0}S\big),
\end{align}
and
\begin{align}\label{eq:Q2def}
\begin{split}
\pazocal{Q}_2(E)
&:=\Tr\Big(\big([\pazocal{A}_{c}]_{1\!:\!N,0}^\top[\pazocal{P}']_{1\!:\!N,0} +[\pazocal{P}']_{0,1\!:\!N}[\pazocal{A}_{c}]_{1\!:\!N,0}\big)S\Big).
\end{split}
\end{align}

\smallskip
\noindent\emph{Step 1: Lower bound the leading term.}
The induced (Frobenius) norm of the linear functional $\pazocal{Q}_1$ satisfies
\begin{align*}
\|\pazocal{Q}_1\|&=\sup_{\normF{E}=1}|\pazocal{Q}_1(E)|  =2\big\|[\pazocal{E}]_{0,0}S\big\|_{\mathrm{F}} \\
&\ge 2\,\lambda_{\min}(S)\,\big\|[\pazocal{E}]_{0,0}\big\|_{\mathrm{F}}.    
\end{align*}
Applying Lemma~\ref{lm:boundSK} gives that
\begin{equation}\label{eq:Q1norm-lb}
\|\pazocal{Q}_1\|
\ge b_2(h)\big\|[\pazocal{E}_N]_{0,0}\big\|_{\mathrm{F}}. 
\end{equation}

\smallskip
\noindent\emph{Step 2: Upper bound the coupling term.}
By \eqref{eq:ABerror},
\(
\|[\pazocal{A}_{c}]_{1\!:\!N,0}\|\le (1+\|K\|)\epsilon(h).
\)
Using Lemma~\ref{lm:traceIneq}, \eqref{eq:Q2def} gives, for all $E$,
\begin{align*}
    |\pazocal{Q}_2(E)|
    &\le 2\|[\pazocal{A}_{c}]_{1\!:\!N,0}\|\,\|\pazocal{P}'\|\,\Tr(S)
    \\
    &\le 2(1+\|K\|)\epsilon(h)\,\|\pazocal{P}'\|\,\Tr(S).    
\end{align*}

Invoking Lemma~\ref{lm:boundKbyJK} (so $\|K\|\le a_1h+a_2\sqrt{h}$),
Lemma~\ref{lm:boundSK} (so $\Tr(S)\le n_x\|S\|\le n_x(h+\epsilon_1)$),
and \eqref{eq:PderivBound} (so $\|\pazocal{P}'\|\le c_P(h)\|E\|$ with
$c_P(h)$ given by \eqref{eq:PderivBound}) gives that
\begin{equation}\label{eq:Q2op-bound}
\|\pazocal{Q}_2\|
:=\sup_{\|E\|_{\mathrm F}=1}|\pazocal{Q}_2(E)|
\le b_3(h)\,\epsilon_1,
\end{equation}

\smallskip
\noindent\emph{Step 3: Conclude.}
From \eqref{eq:grad-decomp},
\[
\|\nabla \hat{\pazocal{J}}_N(K)\|_{\mathrm F}
=\sup_{\|E\|_{\mathrm F}=1}\big|\pazocal{Q}_1(E)+\pazocal{Q}_2(E)\big|
\ge \|\pazocal{Q}_1\|-\|\pazocal{Q}_2\|.
\]
Combining \eqref{eq:Q1norm-lb} and \eqref{eq:Q2op-bound} yields \eqref{eq:gradLowerBound}.
\end{proof}

The $L$-smoothness of the objective function $\hat{\pazocal{J}}_N(K)$ is demonstrated in the below lemma by using the expression for the Hessian matrix in \eqref{eq:Hess}.

\begin{lem}\label{lm:Lsmooth}
Under the conditions of Lemma~\ref{lm:PYLyapuheter}, the objective function
$\hat{\pazocal{J}}_N(K)$ is $L_1(h)$-smooth over the sublevel set
$\pazocal{S}(\bar{\xi},h)$, with
\begin{align}\label{eq:Lsmooth}
\begin{split}
    L_1(h)&=
    \frac{2n_x\|R\|}{\lambda_{\min}(Q)}(h+1) + 8\bar{b}\frac{n_x (h+1)^2}{\lambda_{\min}(Q)^2}\\
    &\quad \times \Big((a_1\|R\|+\bar{b})h + a_2\|R\|\sqrt{h}+\bar{b}\Big),
\end{split}
\end{align}
where $\bar{b}:=\sup_{\xi\in\Xi}\|B(\xi)\|$.
\end{lem}

\begin{proof}
Over the sublevel set $\pazocal{S}(\bar{\xi},h)$, Lemma~\ref{lm:PYLyapuheter} implies
\begin{align}\label{eq:PYpozocalNorm}
\|\pazocal{P}_N(K)\|\le h+\epsilon_1.
\end{align}
Applying Lemma~\ref{lm:LyaEquaIntegral} to \eqref{eq:LyapunovYPCE} and using the cyclic
property of the trace yields
\begin{align}\label{eq:traceYPCEBound}
\begin{split}
    &\Tr(\pazocal{Y}_N(K))
    =
    \Tr\!\left(
        \int_{0}^{\infty}
        e^{\pazocal{A}_{N,c}^\top t}
        e^{\pazocal{A}_{N,c} t}
        \,\de t\,
        \pazocal{I}_N\pazocal{I}_N^\top
    \right) \\
    &\quad \le
    \Big\|
        \int_{0}^{\infty}
        e^{\pazocal{A}_{N,c}^\top t}
        e^{\pazocal{A}_{N,c} t}
        \,\de t
    \Big\|\, n_x  \le
    \frac{n_x(h+\epsilon_1)}{\lambda_{\min}(Q)},
\end{split}
\end{align}
where the second inequality follows from Lemma~\ref{lm:traceIneq}, and the third from $\int_{0}^{\infty}
e^{\pazocal{A}_{N,c}^\top t}
e^{\pazocal{A}_{N,c} t}
\,\de t
\preceq
\lambda_{\min}(Q)^{-1}\pazocal{P}_N(K)$.

Let $E\in\R^{n_u\times n_x}$ with $\|E\|_{\mathrm F}=1$.
Using Lemma~\ref{lm:traceIneq} and the cyclic property of the trace,
the first term in \eqref{eq:Hess} satisfies
\begin{align}\label{eq:hessbound1}
\begin{split}
    2\Big|
        \langle I_{N+1}\!\otimes\!E,\,
        (I_{N+1}\!\otimes\!RE)\pazocal{Y}_N(K)
        \rangle
    \Big|
    \le
    \frac{2n_x\|R\|}{\lambda_{\min}(Q)}(h+\epsilon_1).
\end{split}
\end{align}

We now bound the second term in \eqref{eq:Hess}. For any $E\in\R^{n_u\times n_x}$ with $\|E\|_{\mathrm F}=1$, by \eqref{eq:PderivBound}, \eqref{eq:traceYPCEBound}, and Lemma~\ref{lm:traceIneq},
\begin{align}\label{eq:hessbound2}
\begin{split}
    &4\Big|
        \langle I_{N+1}\!\otimes\!E,\,
        \pazocal{B}_N^\top
        \pazocal{P}'_N(K)[E]
        \pazocal{Y}_N(K)
        \rangle
    \Big| \\
    &\le
    4\|E\|\,\|\pazocal{B}_N\|\,
    \|\pazocal{P}'_N(K)[E]\|\,
    \Tr(\pazocal{Y}_N(K)) \\
    &\le
    8\bar{b}
    \frac{n_x (h+\epsilon_1)^2}{\lambda_{\min}(Q)^2}
    \Big((a_1\|R\|+\bar{b})h
    + a_2\|R\|\sqrt{h}
    + \bar{b}\epsilon_1\Big).
\end{split}
\end{align}

Since $\epsilon_1\in(0,1)$, combining
\eqref{eq:hessbound1} and \eqref{eq:hessbound2} and applying
\cite[Lemma~1.2.2]{nesterov2013introductory} establishes the
$L_1(h)$-smoothness bound \eqref{eq:Lsmooth}.
\end{proof}

The following lemma states the $\pazocal{K}$-PL condition for $\pazocal{J}_{\bar{\xi}}(K)$ established in~\cite{Cui2024}: the gradient norm is lower bounded by a class-$\pazocal{K}$ comparison function $\alpha$ evaluated at the suboptimality gap.

\begin{lem}[$\pazocal{K}$-PL condition {\cite{Cui2024}}]\label{lm:KPLcondition}
There exists a class-$\pazocal{K}$ function $\alpha$ such that, for any $K \in \pazocal{S}(\bar{\xi})$,
\begin{equation}
    \|\nabla \pazocal{J}_{\bar{\xi}}(K)\|_{\mathrm F}
    \;\ge\;
    \alpha\!\left(
        \pazocal{J}_{\bar{\xi}}(K)
        -
        \pazocal{J}_{\bar{\xi}}(K_{\bar{\xi}}^*)
    \right),
\end{equation}
where $K_{\bar{\xi}}^*$ denotes the optimal feedback gain minimizing $\pazocal{J}_{\bar{\xi}}$. The comparison function $\alpha$ is given by
\begin{equation}
    \alpha(r) = \frac{r}{a_3 r + a_4},
\end{equation}
where $a_3$ and $a_4$ are positive constants depending only on $(A(\bar{\xi}), B(\bar{\xi}), Q, R)$.
\end{lem}

The following lemma shows the $L_2$-Lipchitz continuity of $\nabla \pazocal{J}_{\bar \xi}$.
\begin{lem}[Lemma~5.3 in~\cite{cui2025perturbed}]\label{lm:Lsmoothness}
The gradient $\nabla \pazocal{J}_{\bar{\xi}}(K)$ is $L_2(h)$-Lipschitz continuous over the sublevel set $\pazocal{S}(\bar{\xi}, h)$, with Lipschitz constant
\begin{align}
\begin{split}
    L_2(h)
    &= \frac{2\|R\|}{\lambda_{\min}(Q)}\,h
    + \frac{8\,a_2(\bar{\xi})\,\|B(\bar{\xi})\|\,\|R\|}{\lambda_{\min}(Q)^2}\,h^{5/2} \\
    &\quad
    + \frac{8\,\|B(\bar{\xi})\|\big(a_1(\bar{\xi})\|R\|+\|B(\bar{\xi})\|\big)}{\lambda_{\min}(Q)^2}\,h^{3}.
\end{split}
\end{align}
\end{lem}

\section{Supporting Lemmas for Theorem \ref{thm:PCEApproConv}}

The following lemma establishes exponential stability of the closed-loop system
$(A(\xi),B(\xi))$ under any gain $K\in\pazocal{S}(\bar{\xi},h)$.

\begin{lem}\label{lm:PhiBound}
Let
\[
A_c(K,\xi):=A(\xi)-B(\xi)K,
\qquad
\Phi(t,\xi,K):=e^{A_c(K,\xi)t}.
\]
Under Assumption~\ref{ass:heterogeneity}, for any $K\in\pazocal{S}(\bar{\xi},h)$ and any $\xi\in\Xi$, the transition matrix satisfies
\begin{equation}\label{eq:expDecayPhi}
\|\Phi(t,\xi,K)\|\le \mu(h)\,e^{-\gamma(h)t},\qquad \forall t\ge 0,
\end{equation}
where
\begin{align*}
\gamma(h)&:={\lambda_{\min}(Q)}/{(2h)}, \\
\mu(h)&:= \Big(4h\big(\|A(\bar{\xi})\|+\|B(\bar{\xi})\|(a_1h+a_2\sqrt{h})\big)/\eigmin{Q}\Big)^{1/2}.
\end{align*}

\end{lem}

\begin{proof}
Fix $K\in\pazocal{S}(\bar{\xi},h)$. By \eqref{eq:ABKPbarxi} and the perturbation bound
\eqref{eq:Perrbound}, we have
\begin{equation}\label{eq:ABKPbarxi2}
A_c(K,\bar{\xi})^\top P(K,\bar{\xi})+P(K,\bar{\xi})A_c(K,\bar{\xi})
+\tfrac{1}{2}\lambda_{\min}(Q)I_{n_x}\preceq 0 .
\end{equation}
Applying Lemma~\ref{lem:min-eig-PY} to \eqref{eq:ABKPbarxi2} yields
\[
\lambda_{\min}\!\big(P(K,\bar{\xi})\big) 
\ge
\frac{\lambda_{\min}(Q)}{4\big(\|A(\bar{\xi})\|+\|B(\bar{\xi})\|\|K\|\big)} .
\]
Using Lemma~\ref{lm:boundKbyJK}, $\|K\|\le a_1h+a_2\sqrt{h}$ for all
$K\in\pazocal{S}(\bar{\xi},h)$, and $\|P(K,\bar{\xi})\|\le \Tr(P(K,\bar{\xi}))\le h$ gives that the stated upper bound on $\mu(h)$. Finally, Lemma~\ref{lm:ExpDecay} applied to \eqref{eq:ABKPbarxi2} gives \eqref{eq:expDecayPhi}.
\end{proof}

The below lemma shows that all parameter derivatives of the closed-loop transition matrix $\Phi(t,\xi)$ are bounded by a polynomial in $t$ times the same exponential decay rate $e^{-\gamma t}$, uniformly over $K \in \pazocal{S}(\bar{\xi}, h)$ and $\xi \in \Xi$.

\begin{lem}\label{lm:PhiDerivBound}
Under Assumption \ref{ass:heterogeneity}, for any fixed $K\in\pazocal{S}(\bar{\xi},h)$ and any $\xi\in\Xi$, the $i$th parameter
derivative $\Phi^{(i)}(t,\xi):=\frac{\partial^i}{\partial \xi^i}\Phi(t,\xi)$ satisfies
\begin{equation}\label{eq:PhiBound}
\|\Phi^{(i)}(t,\xi)\|
\;\le\;
\sum_{j=0}^{i} d_{i,j}(h)\, t^{j}\,e^{-\gamma t},\qquad \forall t\ge 0,
\end{equation}
for some coefficients $d_{i,j}(h)\ge 0$ depending on $h$ (and $i$).
\end{lem}

\begin{proof}
In the proof, the dependence on $K$ is omitted whenever no confusion arises. Since $A(\xi)$ and $B(\xi)$ are analytic on the compact set $\Xi$, for each $i\ge 0$ define
\[
\bar c_i:=\max_{\xi\in\Xi}\|A^{(i)}(\xi)\|,
\qquad
\bar d_i:=\max_{\xi\in\Xi}\|B^{(i)}(\xi)\|.
\]
Then, for any $\xi\in\Xi$,
\[
\|A_c^{(i)}(\xi)\|
=\|A^{(i)}(\xi)-B^{(i)}(\xi)K\|
\le f_i(h),
\]
where according to Lemma~\ref{lm:boundKbyJK},
\[
f_i(h):= \bar c_i+\bar d_i\big(a_1h+a_2\sqrt{h}\big),
\]

We prove \eqref{eq:PhiBound} by induction on $i$.

\emph{Base case $i=0$.}
By setting $d_{0,0}(h)=\mu(h)$, the bound \eqref{eq:PhiBound} holds for $i=0$ as a direct consequence of Lemma~\ref{lm:PhiBound}.

\emph{Induction step.}
Assume \eqref{eq:PhiBound} holds for all orders $0,1,\ldots,i-1$. Differentiating
$\dot{\Phi}=A_c(\xi)\Phi$ $i$ times with respect to $\xi$ and applying Leibniz's rule gives
\begin{align*}
\frac{\partial}{\partial t}\Phi^{(i)}(t,\xi) &= A_c(\xi)\Phi^{(i)}(t,\xi)
+\sum_{j=1}^{i}\binom{i}{j}A_c^{(j)}(\xi)\,\Phi^{(i-j)}(t,\xi)\\
\Phi^{(i)}(0,\xi) &= 0.    
\end{align*}
Hence
\[
\Phi^{(i)}(t,\xi)
=\sum_{j=1}^{i}\binom{i}{j}\int_0^t
\Phi(t-s,\xi)\,A_c^{(j)}(\xi)\,\Phi^{(i-j)}(s,\xi)\,\de s.
\]
Taking norms, using $\|\Phi(t-s,\xi)\|\le \mu(h)e^{-\gamma(t-s)}$, $\|A_c^{(j)}(\xi)\|\le f_j(h)$, and the induction hypothesis for $\Phi^{(i-j)}$ yields
\begin{align*}
\|\Phi^{(i)}(t,\xi)\| &\le \mu(h)\sum_{j=1}^{i}\binom{i}{j} f_j(h) \\
&\quad \times \sum_{k=0}^{i-j} d_{i-j,k}(h)\int_0^t e^{-\gamma(t-s)} s^k e^{-\gamma s}\,\de s.    
\end{align*}
Note that 
\[
\int_0^t e^{-\gamma(t-s)} s^k e^{-\gamma s}\,\de s
= \frac{t^{k+1}}{k+1}e^{-\gamma t}.
\]
Therefore,
\[
\|\Phi^{(i)}(t,\xi)\|
\le
\sum_{k=0}^{i-1}\Bigg(
\sum_{j=1}^{i-k}\binom{i}{j}\frac{\mu(h)\,f_j(h)\,d_{i-j,k}(h)}{k+1}
\Bigg)t^{k+1}e^{-\gamma t}.
\]
This establishes \eqref{eq:PhiBound} for order $i$ with the recursion
\begin{align*}
d_{i,k+1}(h)&:=\sum_{j=1}^{i-k}\binom{i}{j}\frac{\mu(h)\,f_j(h)\,d_{i-j,k}(h)}{k+1}, \\
k&=0,\ldots,i-1,
\end{align*}
and $d_{i,0}(0) :=0$.
\end{proof}

\section*{References}

\def\refname{\vadjust{\vspace*{-2.5em}}} 

\bibliographystyle{ieeetr}
\bibliography{references}  
\end{document}